\documentclass[leqno,a4paper,twoside,11pt]{article}
\usepackage[nointlimits,nosumlimits]{amsmath}
\usepackage{amsfonts,amssymb,amsthm,ifthen}
\usepackage{color}
\usepackage[arrow,matrix,curve]{xy}
\usepackage[utf8x]{inputenc}
\usepackage[square,numbers]{natbib}

\newtheorem{Thm}{Theorem}[section]
\newtheorem{Prp}[Thm]{Proposition}
\newtheorem{Lem}[Thm]{Lemma}

\newtheorem{Def}[Thm]{Definition}
\newtheorem{Rem}[Thm]{Remark}
\newtheorem{Exp}[Thm]{Example}

\theoremstyle{definition}

\newcommand{\arr}[2]{%
  \begin{array}{@{}#1@{}}#2\end{array}}
\newcommand{\abs}[1]{\left| #1 \right|}
\newcommand{\scal}[3][]{\ifthenelse{\equal{#1}{}}{
  \left\langle #2,\,#3 \right\rangle
}{\ifthenelse{\equal{#1}{(}}{
  \left( #2,\,#3 \right)
}{\ifthenelse{\equal{#1}{[}}{
  \left[ #2,\,#3 \right]
}{
  #1\left( #2,\,#3 \right)
}}}}
\newcommand{\dd}[2][]{\frac{\partial #1}{\partial #2}}
\newcommand{\phantrix}[3]{#1^{#2}_{\phantom{#2}#3}}
\newcommand{\setsep}{\;\big|\;}

\newcommand{\mA}{\mathcal A}
\newcommand{\mB}{\mathcal B}
\newcommand{\mE}{\mathcal E}
\newcommand{\mF}{\mathcal F}
\newcommand{\mN}{\mathcal N}
\newcommand{\mO}{\mathcal O}
\newcommand{\mP}{\mathcal P}
\newcommand{\mS}{\mathcal S}

\newcommand{\bN}{\mathbb N}
\newcommand{\bR}{\mathbb R}
\newcommand{\bZ}{\mathbb Z}
\newcommand{\fg}{\mathfrak g}
\newcommand{\fgl}{\mathfrak{gl}}
\newcommand{\oP}{\overline{P}}

\newcommand{\Der}{\mathrm{Der}}
\newcommand{\End}{\mathrm{End}}
\newcommand{\Gr}{\mathrm{Gr}}
\newcommand{\Hol}{\mathrm{Hol}}
\newcommand{\Hom}{\mathrm{Hom}}
\newcommand{\Mat}{\mathrm{Mat}}
\newcommand{\ev}{\mathrm{ev}}
\newcommand{\hol}{\mathrm{hol}}
\newcommand{\id}{\mathrm{id}}
\newcommand{\rk}{\mathrm{rk}}
\newcommand{\tr}{\mathrm{tr}}

\renewcommand{\title}[1]{\vbox{\center\LARGE{\textsc{#1}}}\vspace{5mm}}
\renewcommand{\author}[1]{\vbox{\center\large{\textsc{#1}}}\vspace{5mm}}
\newcommand{\address}[1]{\vbox{\center\em#1}}
\newcommand{\email}[1]{\vbox{\center\tt#1}\vspace{5mm}}

\begin{document}

\title{Super Wilson Loops and\\Holonomy on Supermanifolds}

\author{Josua Groeger$^1$}

\address{Humboldt-Universit\"at zu Berlin, Institut f\"ur Mathematik,\\
  Rudower Chaussee 25, 12489 Berlin, Germany }

\email{$^1$groegerj@mathematik.hu-berlin.de}

\begin{abstract}
\noindent
The classical Wilson loop is the gauge-invariant trace of the parallel transport
around a closed path with respect to a connection on a vector bundle over a smooth manifold.
We build a precise mathematical model of the super Wilson loop,
an extension introduced by Mason-Skinner and Caron-Huot, by endowing the objects
occurring with auxiliary Graßmann generators coming from S-points.
A key feature of our model is a supergeometric parallel transport,
which allows for a natural notion of holonomy on a supermanifold
as a Lie group valued functor. Our main results for that theory
comprise an Ambrose-Singer theorem as well as a natural analogon of the holonomy principle.
Finally, we compare our holonomy functor with the holonomy supergroup introduced
by Galaev in the common situation of a topological point.
It turns out that both theories are different, yet related in a sense made precise.
\end{abstract}

\section{Introduction}

Gluon scattering amplitudes have been known to be dual to Wilson loops along lightlike polygons
(\cite{AM07,DKS08,BHT08,AR08}).
While these quantum expectation values, which are formally calculated by means of the path integral,
remain problematic from a mathematical point of view, the underlying classical theory has been well understood.
In fact, a Wilson line refers to parallel transport with respect to a connection on a vector bundle along
a path in the underlying smooth manifold. In the usual context of flat spacetime (Minkowski space)
with a single global coordinate chart, the corresponding solution operator can be written in terms of
a path-ordered exponential.

Recently, a similar duality (at weak coupling) between the full superamplitude of
$\mN=4$ super Yang-Mills theory and two variants of a supersymmetric extension of the Wilson loop
has been claimed. The first approach (\cite{MS10}) originates in momentum twistor
space and translates into the integral over a superconnection in spacetime, while
the second (\cite{CH11}) attaches to lightlike polygons certain edge and vertex operators, whose
shape is determined by supersymmetry constraints (\cite{Gro12}).
Both approaches agree, in the common domain of definition,
up to a term depending on the equations of motion (\cite{BKS12})
and indeed satisfy the conjectured duality upon subtracting an anomalous contribution (\cite{Bel12}).

The first purpose of the present article is to
build a supergeometric model of super Wilson loops
that leads to the same characteristic formulas as summarised in Sec. 2.2 of \cite{BKS12}.
The main idea is to give the objects occurring an inner structure through auxiliary Graßmann
generators coming from $S$-points. While the resulting additional degrees of freedom come
without physical significance, this approach is well-justified mathematically and has been performed
successfully in modelling other aspects of superfield theory.
Notably, consider ''maps with flesh'' as introduced by H{\'{e}}lein in \cite{Hel09}
as models for superfields including bosons and fermions. See also
\cite{DF99b,Khe07,Han12} for the same concept under different terminology
and \cite{Gro11a} for their differential calculus.

A key feature of our model is the supergeometric parallel transport introduced in
Sec. \ref{secSuperWilsonLoops}, which allows for a natural notion of holonomy
at an $S$-point of a supermanifold as a Lie group valued functor.
A different notion of holonomy on supermanifolds was introduced by Galaev in \cite{Gal09} by
taking a suitable generalisation of the Ambrose-Singer theorem as the definition of a
super Lie algebra and endowing this to a Harish-Chandra pair, thus obtaining a super
Lie group for every topological point of the manifold.
Developing a new holonomy theory by means of our parallel transport, and comparing it to Galaev's,
is the second objective of this article.

In Sec. \ref{secHolonomy}, we establish two main results generalising properties of classical
holonomy. The first is an Ambrose-Singer theorem, which describes the holonomy Lie algebra
in terms of curvature, while the second formulates a natural analogon of the holonomy principle
relating parallel sections to holonomy-invariant vectors.

Our Ambrose-Singer theorem facilitates the comparison of our holonomy functor
with Galaev's theory, which is the subject matter of Sec. \ref{secComparison}.
Since this functor is, in general, not representable,
both theories are different in the common situation of a topological point.
Nevertheless, we show that they are related in that
the generators of Galaev's holonomy algebra can be
extracted as certain coefficients by considering special $S$-points.
This construction is based on the knowledge of the geometric significance of the
elements and, in this sense, not algebraic.

\section{Super Wilson Loops and Parallel Transport}
\label{secSuperWilsonLoops}

The super Wilson loop described in \cite{MS10} and \cite{BKS12} is constructed as follows.
Consider $n$ ''superpoints'' $(x_i,\theta_i)$ in chiral superspace,
which are symbolic quantities in that
their exact mathematical type is not important, only their calculation rules such as
\begin{align}
\label{eqnCalculationRules}
x_i^{\mu}\cdot x_j^{\nu}=x_j^{\nu}\cdot x_i^{\mu}\;,\qquad
\theta_i^{\alpha A}\cdot\theta_j^{\beta B}=-\theta_j^{\beta B}\cdot\theta_i^{\alpha A}
\end{align}
These superpoints are connected by ''straight lines''
\begin{align}
\label{eqnStraightSuperline}
x(t_i)=x_i-t_ix_{i,i+1}\;,\qquad\theta(t_i)=\theta_i-t_i\theta_{i,i+1}
\end{align}
thus yielding a closed ''superpath'' $\gamma$
parametrised by one bosonic variable $t$, which enters the Wilson loop via
\begin{align}
\label{eqnWilsonLine}
W_{\gamma}=\tr\left(X\mapsto\mP\exp\left(\int_0^1ig\mB(t)dt\right)[X]\right)\;,\qquad\mB(t)=\mB_{\xi}\cdot\dot{\gamma}^{\xi}(t)
\end{align}
where $\mP\exp$ denotes a path-ordered exponential, and $\mB_{\xi}$ is a connection one-form in coordinates $\xi$.
This connection has a very specific form due to supersymmetry conditions which,
however, is not relevant for our purposes.

As a mathematical model for a more general situation,
let $M$ be a supermanifold of dimension $\dim M=(\dim M)_{\overline{0}}|(\dim M)_{\overline{1}}$
(such as chiral superspace), and let $S$ be
another supermanifold which should be thought of as auxiliary. Throughout,
we employ the definitions of Berezin-Kostant-Leites \cite{Lei80}.
A supermanifold $M$ is thus, in particular, a ringed space $M=(M_0,\mO_M)$,
and a morphism $\varphi:M\rightarrow N$ consists of two parts $\varphi=(\varphi_0,\varphi^{\sharp})$
with $\varphi_0:M_0\rightarrow N_0$ a smooth map and $\varphi^{\sharp}$ a generalised pullback
of superfunctions $f\in\mO_N$. Modern monographs on the general theory of supermanifolds include
\cite{Var04} and \cite{CCF11}.

\begin{Def}
\label{defSPoint}
An $S$-point of $M$ is a morphism $x=(x_0,x^{\sharp}):S\rightarrow M$.
A (smooth) $S$-path $\gamma$ connecting $S$-points $x$ and $y$ is a morphism
\begin{align*}
\gamma=(\gamma_0,\gamma^{\sharp}):S\times [0,1]\rightarrow M\qquad\textrm{such that}\qquad
\ev|_{t=0}\gamma^{\sharp}=x^{\sharp}\;,\qquad\ev|_{t=1}\gamma^{\sharp}=y^{\sharp}
\end{align*}
which we shall denote, by a slight abuse of notation, by $\gamma:x\rightarrow y$.
It is called closed (or an $S$-loop) if $x=y$.
\end{Def}

In the following, we will exclusively consider superpoints
\begin{align}
\label{eqnSuperpoint}
S=\bR^{0|L}=\left(\{0\},\bigwedge\bR^L\right)\;,\qquad\bigwedge\bR^L=\langle\eta^1,\ldots,\eta^L\rangle\;,\qquad L\in\bN
\end{align}
Although most of our results should continue to hold accordingly for general $S$, this restriction will
turn out to suffice for reproducing the characteristic formulas of super Wilson loops as well as
allowing for a powerful holonomy theory.
This significance of superpoints does not come unexpected. According to \cite{SW11}, an inner Hom object
$\underline{\Hom}(M,N)$ in the category of supermanifolds is determined by its $\bigwedge\bR^L$-points
\begin{align*}
\underline{\Hom}(M,N)\left(\bigwedge\bR^L\right)\cong\scal[\Hom_{\mathrm{SMan}}]{\bR^{0|L}\times M}{N}
\end{align*}
in the sense of Molotkov-Sachse theory (\cite{Mol84,Sac08}).
The morphisms on the right are the aforementioned ''maps with flesh'' (\cite{Hel09}).

\begin{Def}
\label{defPathConcatenation}
Let $x,y,z:S\rightarrow M$ be $S$-points and $\gamma:x\rightarrow y$ and $\delta:y\rightarrow z$ be 
$S$-paths. For fixed $t_0\in[0,1]$, we prescribe
\begin{align*}
\ev|_{t=t_0}(\delta\star\gamma)^{\sharp}:=\left\{\arr{ll}{\ev|_{t=2t_0}\gamma^{\sharp}&t_0\leq 1/2\\\ev|_{t=2\left(t_0-\frac{1}{2}\right)}\delta^{\sharp}&t_0\geq 1/2}\right.
\end{align*}
This defines an $S$-point which coincides with $x$,$y$,$z$ for $t_0=0,\frac{1}{2},1$, respectively.
Similarly, we define
\begin{align*}
\ev|_{t=t_0}(\gamma^{-1})^{\sharp}:=\ev|_{t=(1-t_0)}\gamma^{\sharp}
\end{align*}
\end{Def}

Considering all $t_0\in[0,1]$ at a time, the previous definition yields $S$-paths
$\delta\star\gamma:x\rightarrow z$ and $\gamma^{-1}:y\rightarrow x$,
referred to as the concatenation of $\gamma$ and $\delta$ and the inverse of $\gamma$,
respectively. The concatenation is, however,
only piecewise smooth in the sense of the following definition.

\begin{Def}
Let $x,y$ be $S$-points. A piecewise smooth $S$-path $\gamma:x\rightarrow y$ connecting $x$ with $y$
is a tuple $(\gamma_j:S\times[t_j,t_{j+1}]\rightarrow M)_{j=0}^l$
with $t_0=0$, $t_l=1$ and $t_j<t_{j+1}$
such that $\ev|_{t=t_{j+1}}\gamma^{\sharp}_j=\ev|_{t=t_{j+1}}\gamma^{\sharp}_{j+1}$ and $\ev|_{t=0}\gamma^{\sharp}_0=x^{\sharp}$
and $\ev|_{t=1}\gamma^{\sharp}_l=y^{\sharp}$, and such that $\gamma_j|_{S\times[t_j,t_{j+1}]}$ is a morphism.
\end{Def}

The concatenation $(\delta\star\gamma)$ and inverse $\gamma^{-1}$ of piecewise smooth
paths $\delta$ and $\gamma$
are defined analogous. The construction is such that the underlying path $(\delta\star\gamma)_0$
is the classical concatenation of $\delta_0$ and $\gamma_0$, and $(\gamma^{-1})_0=(\gamma_0)^{-1}$.

\begin{Exp}
\label{expSPoint}
Comparing with the objects in \cite{BKS12}, we state the following dictionnary. Let
$(x^{\mu},\theta^{\alpha A})$ denote (global) coordinates on $M\cong\bR^{n|m}$ (using
spacetime indices $\mu$ rather than spinor indices $\dot{\alpha}\alpha$).
Then a superpoint is an $S$-point $\xi=(\xi_0,\xi^{\sharp}):S\rightarrow M$, identified
with $(\xi^{\sharp}(x^{\mu}),\xi^{\sharp}(\theta^{\alpha A}))\in(\mO_S)^{n|m}$.
The latter tuple is then abbreviated $(x,\theta)=(x^{\mu},\theta^{\alpha A})$, for which (\ref{eqnCalculationRules}) is satisfied.
The straight line connecting superpoints $(x_i,\theta_i)$ and $(x_{i+1},\theta_{i+1})$ is the
$S$-path $\xi_{i,i+1}:S\times[0,1]\rightarrow M$ defined as follows.
\begin{align*}
&(\xi^{\sharp}_{i,i+1}(x^{\mu}),\xi^{\sharp}_{i,i+1}(\theta^{\alpha A}))\\
&\qquad:=\left(\xi_i^{\sharp}(x^{\mu})-t\left(\xi_i^{\sharp}(x^{\mu})-\xi_{i+1}^{\sharp}(x^{\mu})\right),\,
\xi_i^{\sharp}(\theta^{\alpha A})-t\left(\xi_i^{\sharp}(\theta^{\alpha A})-\xi_{i+1}^{\sharp}(\theta^{\alpha A})\right)\right)\\
&\qquad\in(\mO_{S\times[0,1]})^{n|m}
\end{align*}
In this sense, we can understand (\ref{eqnStraightSuperline}).
The last line is $\xi_{n,0}$. Concatenation thus yields a loop.
\end{Exp}

\subsection{Super Vector Bundles and Connections}

A super vector bundle $\mE$ over a supermanifold $M$ is a sheaf of locally free $\mO_M$ supermodules
on $M$. We shall denote its even and odd parts by $\mE_{\overline{0}}$ and $\mE_{\overline{1}}$,
respectively. An important example is the super tangent bundle $\mS M:=\Der(\mO_M)$, which is
the sheaf of $\mO_M$-superderivations.
$\mE(U)$ is, for $U\subseteq M_0$ sufficiently small, by definition isomorphic to $\mO_M(U)^{\rk\mE}$
with $\rk\mE=(\rk\mE)_{\overline{0}}|(\rk\mE)_{\overline{1}}$ the rank of $\mE$.
Let $(T^j)^{\rk\mE}_{j=1}$ be an adapted local basis such that
$X\in\mE(U)$ is identified with the tuple $(X^j)_{j=1}^{\rk\mE}$ of functions $X^j\in\mO_M(U)$ with
respect to right coefficients $X=T^j\cdot X^j$ (sum convention).
In general, it is preferable to consider right coordinates on supermodules over
supercommutative superalgebras, for then superlinear maps can be identified with matrices.
For example, the matrix of the differential $d\varphi[X]:=X\circ\varphi^{\sharp}$
for $X=\partial_{\xi^k}\cdot X^k\in\mS M$ of a map $\varphi:M\rightarrow N$ with respect to
coordinates $(\xi^j)$ and $(\zeta^j)$ is given by
\begin{align}
\label{eqnDifferentialMatrix}
d\varphi^i_{\phantom{i}k}:=(-1)^{(\abs{\xi^k}+\abs{\zeta^i})\cdot\abs{\zeta^i}}\dd[\varphi^{\sharp}(\zeta^i)]{\xi^k}\quad
\mathrm{s.th.}\quad d\varphi[X]=\sum_{i,k}\left(\varphi^{\sharp}\circ\dd{\zeta^i}\right)\cdot d\varphi^i_{\phantom{i}k}\cdot X^k
\end{align}

\begin{Lem}[Chain Rule]
\label{lemChainRule}
Let $\varphi:M\rightarrow N$ and $\psi:N\rightarrow P$ be morphisms. Then
\begin{align*}
d(\psi\circ\varphi)[X]
=\left(\varphi^{\sharp}\circ\psi^{\sharp}\circ\dd{\pi^l}\right)\cdot\varphi^{\sharp}(d\psi^l_{\phantom{l}i})\cdot d\varphi^i_{\phantom{i}k}\cdot X^k
\end{align*}
with $(\pi^l)$ coordinates on $P$ and indices $k$, $i$ referring to (unlabelled)
coordinates on $M$ and $N$, respectively.
\end{Lem}

\begin{proof}
This is proved by a straightforward calculation in local coordinates.
\end{proof}

\begin{Def}
\label{defSConnection}
For a super vector bundle $\mE$, and $S$ as in (\ref{eqnSuperpoint}),
we define
\begin{align*}
\mE_S:=\mE\otimes_{\mO_M}\mO_{S\times M}
\end{align*}
An $S$-connection on $M$ is an even $\bR$-linear sheaf morphism
\begin{align*}
\nabla:\mE_S\rightarrow\mS M^*_S\otimes_{\mO_{S\times M}}\mE_S\;,\qquad
\nabla(fe)=df\otimes_{\mO_{S\times M}}e+f\cdot\nabla e\qquad\mathrm{for}\quad f\in\mO_{S\times M}
\end{align*}
\end{Def}

In particular, $\mE_S$ can be considered as a super vector bundle on $S\times M$ and, in this sense,
$\nabla$ is an ordinary superconnection.
The local picture is as follows. Let $\xi=(x,\theta)$ be coordinates on $M$ and
$(T^j)$ be an $\mE$-basis. Then $X\in\mE_S$ can be expaned as $X=T^j\cdot X^j$
with $X^j\in\mO_{S\times M}(\{0\}\times U)$, and
\begin{align}
\label{eqnLocalConnection}
\nabla_{\partial_{\xi^i}}X
=(-1)^{\abs{\xi^i}\abs{T^j}}T^j\partial_{\xi^i}(X^j)+\Gamma_{\xi^i}[T^j]\cdot X^j\;,\qquad
\Gamma_{\xi^i}[T^j]:=\nabla_{\partial_{\xi^i}}T^j
\end{align}
where $\Gamma_{\xi^i}\in\Mat_{\rk\mE\times\rk\mE}(\mO_{S\times M}(\{0\}\times U))$,
which has an expansion
\begin{align*}
\Gamma_{\xi^i}=\sum_{I=(i_1,\ldots,i_{\abs{I}})}\theta^I\cdot(\Gamma_{\xi^i})_I\;,\qquad
(\Gamma_{\xi^i})_I\in\Mat_{\rk\mE\times\rk\mE}(\mO_{S\times M_0}(\{0\}\times U))
\end{align*}

\begin{Exp}
\label{expBKSConnection}
Consider the trivial vector bundle $\mE:=\mathfrak{su}(N)\otimes_{\bR}\mO_M$  with $N\in\bN$ of rank $\rk\mE=\dim\mathfrak{su}(N)|0$
over flat superspace with global coordinates $\xi=(x^{\mu},\theta^{\alpha A})$. Define
$\mA_{\mu}:=\Gamma_{x^{\mu}}$ and $\mF_{\alpha A}:=\Gamma_{\theta^{\alpha A}}$.
With this notation, the $\theta$-expansion assumes the form
\begin{align*}
\mA_{\mu}&=(\mA_{\mu})_0+\theta^{\beta B}(\mA_{\mu})_{\beta B}+\theta^{\beta B}\theta^{\gamma C}(\mA_{\mu})_{\beta B\,\gamma C}+\ldots\\
\mF_{\alpha A}&=(\mF_{\alpha A})_0+\theta^{\beta B}(\mF_{\alpha A})_{\beta B}+\theta^{\beta B}\theta^{\gamma C}(\mF_{\alpha A})_{\beta B\,\gamma C}+\ldots
\end{align*}
Since $\nabla$ is, by definition, even it follows that $\mA_{\mu}$ and $\mF_{\alpha A}$ are
even respectively odd. The parity of the $\theta$-coefficients in the expansion is thus
alternating. This is the situation considered in \cite{BKS12}. In case of a plain connection
on $\mE$, the odd coefficients in the $\mA_{\mu}$-expansion would be missing,
and analogous for $\mF_{\alpha A}$.
\end{Exp}

Let $\mE\rightarrow N$ be a super vector bundle over $N$ and $\varphi:M\rightarrow N$ be a morphism
of supermanifolds. The pullback of $\mE$ under $\varphi$ is defined as
\begin{align}
\label{eqnPullbackBundle}
\varphi^*\mE(U):=\mO_M(U)\otimes_{\varphi}(\varphi_0^*\mE)(U)\;,\qquad U\subseteq M_0\quad\mathrm{open}
\end{align}
Here, $\varphi_0^*\mE$ is the pullback of the sheaf $\mE$ under the continuous map $\varphi_0$
which, in terms of its sheaf space, is the bundle of stalks $\mE_{\varphi_0(x)}$ attached to
$x\in M_0$. In this context, one can define the pullback $\varphi_0^*X\in\varphi_0^*\mE$
of $X\in\mE$. (\ref{eqnPullbackBundle}) indeed yields a super vector bundle on $M$ of rank $\rk\mE$.
For details, consult \cite{Han12} and \cite{Ten75}.

A local frame $(T^k)$ of $\mE$ gives rise to a local frame $(\varphi_0^*T^k)$ of $\varphi_0^*\mE$
and a local frame $(\varphi^*T^k:=1\otimes_{\varphi}\varphi_0^*T^k)$ of $\varphi^*\mE$ such that,
locally, every section $X\in\varphi^*\mE$ can be written $X=\varphi^*T^k\cdot X^k$ with
$X^k\in\mO_M(U)$. For $Y=T^kY^k\in\mE$, we find
\begin{align*}
\varphi^*Y=\varphi^*(T^kY^k)=\varphi^*T^k\cdot\varphi^{\sharp}(Y^k)
\end{align*}

\begin{Def}
\label{defCanonicalInclusion}
In the following, we shall identify maps $\varphi:S\times M\rightarrow N$ with maps
$\hat{\varphi}:S\times M\rightarrow S\times N$ by composing $\varphi$ with the canonical inclusion
$N\hookrightarrow S\times N$.
\end{Def}

In particular, we will use this identification for $S$-points $x:S\rightarrow M$ and
$S$-paths $\gamma:S\times[0,1]\rightarrow M$.
In terms of generators $\eta^j$ as in (\ref{eqnSuperpoint}), the construction
is such that $\hat{\varphi}^{\sharp}(\eta^j)=\eta^j$.

\begin{Lem}
\label{lemSPullbackBundle}
Let $\varphi:S\times M\rightarrow N$ and $\mE\rightarrow N$ be a super vector bundle.
Then $\varphi^*\mE\cong\hat{\varphi}^*\mE_S$.
\end{Lem}

Locally, this isomorphism is such that
$X=(\hat{\varphi}^*T^k)\cdot X^k\in\hat{\varphi}^*\mE_S$ is identified with
$X=\varphi^*T^k\cdot X^k\in\varphi^*\mE$.
We define the pullback of $X\in\mE_S$ under $\varphi:S\times M\rightarrow N$ by
\begin{align}
\label{eqnSectionPullback}
\varphi^*X:=\hat{\varphi}^*X\in\hat{\varphi}^*\mE_S\cong\varphi^*\mE
\end{align}
Similarly, an endomorphism $E\in\End_{\mO_{S\times N}}(\mE_S)$ is pulled back under $\varphi$
to an endomorphism along $\varphi$ as follows.
\begin{align}
\label{eqnEndoPullback}
E_{\varphi}\in\End_{\mO_{S\times M}}(\hat{\varphi}^*\mE_S)\;,\qquad
E_{\varphi}(\varphi^*Y):=\varphi^*E(Y)
\end{align}
and analogous for other tensors.

Let $\nabla$ be a connection on $\mE\rightarrow N$ and $\varphi:M\rightarrow N$ be a morphism.
There are two types of pullback connections. With respect to coordinates $(\xi^k)$ of $M$,
we write $X=(\varphi^*\partial_{\xi^i})\cdot X^i\in\varphi^*\mS N$ and prescribe
\begin{align}
\label{eqnPullbackConnectionIntermediate}
&(\varphi^*\nabla):\varphi_0^*\mE\rightarrow(\varphi^*\mS N)^*\otimes_{\mO_M}\varphi^*\mE\\
&(\varphi^*\nabla)_{(\varphi^*\partial_i)X^i}(\varphi^*Z):=(-1)^{\abs{X^i}\abs{\partial_i}}X^i\cdot\varphi^*(\nabla_{\partial_i}Z)\nonumber
\end{align}
The local representations glue together to a well-defined object satisfying a Leibniz rule.
For the second, more common, pullback note that $X\in\varphi^*\mS N$ acts naturally on sections
$f\in\mO_N$ as the superderivation $X(f):=(-1)^{\abs{X^i}\abs{f}}(\varphi^{\sharp}\circ\partial_{\xi^i})(f)\cdot X^i$
along $\varphi$. We define
\begin{align}
\label{eqnPullbackConnectionFinal}
&(\varphi^*\nabla):\varphi^*\mE\rightarrow\mS M^*\otimes_{\mO_M}\varphi^*\mE\\
&(\varphi^*\nabla)_X((\varphi^*T^k)Z^k):=(-1)^{\abs{X}\abs{T^k}}(\varphi^*T^k)\cdot X(Z^k)
+(\varphi^*\nabla)_{d\varphi[X]}\varphi^*T^k\cdot Z^k
\nonumber
\end{align}
using (\ref{eqnDifferentialMatrix}) and (\ref{eqnPullbackConnectionIntermediate})
for the second summand.
Again, this prescription is independent of coordinates and $\mE$-bases and yields a
connection on $\varphi^*\mE\rightarrow M$.

Let now $\nabla$ be an $S$-connection on $\mE_S$ over $N$ and $\varphi:S\times M\rightarrow N$.
We may consider $\nabla$ as an ordinary connection over $S\times N$ and apply
(\ref{eqnPullbackConnectionIntermediate}) to obtain
\begin{align*}
(\hat{\varphi}^*\nabla):\hat{\varphi}_0^*\mE_S\rightarrow\left(\hat{\varphi}^*\mS(S\times N)\right)^*
\otimes_{\mO_{S\times M}}\hat{\varphi}^*\mE_S
\end{align*}
Concatenating this with the adjoint of the inclusion $\mS N_S\subseteq\mS(S\times N)$, and
using $\hat{\varphi}_0^*\mE_S=\varphi_0^*\mE\otimes\mO_S$ as well as
$\hat{\varphi}^*\mS N_S=\varphi^*\mS N$, we yield the first pullback, denoted
\begin{align}
\label{eqnPullbackConnectionIntermediateRelative}
(\varphi^*\nabla):\varphi_0^*\mE\otimes\mO_S\rightarrow(\varphi^*\mS N)^*\otimes_{\mO_{S\times M}}\varphi^*\mE
\end{align}
The second pullback is the connection
\begin{align}
\label{eqnPullbackConnectionRelative}
(\varphi^*\nabla):\varphi^*\mE\rightarrow\mS M_S^*\otimes_{\mO_{S\times M}}\varphi^*\mE
\end{align}
defined verbatim to (\ref{eqnPullbackConnectionFinal}) by means of
(\ref{eqnPullbackConnectionIntermediateRelative})
The local picture is as follows.
\begin{align}
\label{eqnPullbackConnection}
(\varphi^*\nabla)_XZ=(-1)^{\abs{X}\abs{T^k}}(\varphi^*T^k)X(Z^k)
+X(\varphi^*(\xi^l))\hat{\varphi}^*(\nabla_{\partial_{\xi^l}}T^k)\cdot Z^k
\end{align}

\subsection{Parallel Transport}

\begin{Def}
A section $X\in\gamma^*\mE$ is called parallel if $(\gamma^*\nabla)_{\partial_t}X\equiv 0$.
\end{Def}

The local form is as follows. As above, we write $X=(\gamma^*T^k)\cdot X^k$,
thus identifying $X$ with the $t$-dependent column vector $X(t)\in(\mO_S)^{\rk\mE}$.
We further use the notation $\Gamma^m_{lk}\cdot T^m:=\Gamma_{\xi^l}[T^k]$
with $\Gamma_{\xi^l}$ as in (\ref{eqnLocalConnection}).
By (\ref{eqnPullbackConnection}), the parallelness condition in local coordinates reads
\begin{align}
\label{eqnParallel}
\partial_tX(t)=-\mB(t)\cdot X(t)\;,\qquad
\phantrix{\mB(t)}{m}{k}=(-1)^{\abs{T^m}(\abs{T^k}+1)}\partial_t(\gamma^*(\xi^l))\cdot\hat{\gamma}^*(\Gamma^m_{lk})
\end{align}
with $\mB(t)\in\End_{\mO_{S}}(\gamma^*\mE)_{\overline{0}}\cong\Mat_{\rk\mE\times\rk\mE}(\mO_S)_{\overline{0}}$.

\begin{Exp}
In the situation of Exp. \ref{expBKSConnection}, the matrix $\mB(t)$ can be written in the form
\begin{align*}
\mB(t)=\dot{x}^{\mu}(t)\mA_{\mu}+\dot{\theta}^{\alpha A}(t)\mF_{\alpha A}
\end{align*}
This is equation (17) of \cite{BKS12}.
\end{Exp}

The next result follows from standard facts on ODEs applied to (\ref{eqnParallel}).

\begin{Lem}
Let $X_x\in x^*\mE$ be a section along an $S$-point $x:S\rightarrow M$, and
$\gamma$ be a piecewise smooth $S$-path with $\ev|_{t=0}\gamma^{\sharp}=x^{\sharp}$.
Then there exists a unique parallel section $X\in\gamma^*\mE$ along $\gamma$ such that
$\ev|_{t=0}X=X_x$.
\end{Lem}

\begin{Def}
\label{defParallelTransport}
Let $\gamma:x\rightarrow y$ be a smooth $S$-path and
let $X_x\in x^*\mE$ be a vector field along $x:S\rightarrow M$. We define the parallel transport
\begin{align*}
P_{\gamma}:x^*\mE\rightarrow y^*\mE\;,\qquad P_{\gamma}(X_x):=\ev_{t=1}X
\end{align*}
where $X\in\gamma^*\mE$ denotes the parallel vector field such that $\ev_{t=0}X=X_x$.
\end{Def}

For smooth $S$-paths $\gamma:x\rightarrow y$ and $\delta:y\rightarrow z$, we define parallel
transport of the concatenation by $P_{\delta\star\gamma}:=P_{\delta}\circ P_{\gamma}$.
If $\delta\star\gamma$ happens to be smooth, this definition agrees with the one from
Def. \ref{defParallelTransport} by the following lemma.

\begin{Lem}
\label{lemPathConcatenation}
Let $a<b<c$ and $\gamma:S\times[a,c]\rightarrow M$ be a smooth $S$-path. Then
\begin{align*}
P_{\gamma|_{S\times[b,c]}}\circ P_{\gamma|_{S\times[a,b]}}=P_{\gamma}
\end{align*}
\end{Lem}

\begin{proof}
Let $X_x\in x^*\mE$. We define $X\in\gamma^*\mE$ by setting
\begin{align*}
X(t):=P_{\gamma|_{S\times[a,t]}}[X_x]\quad(t\in[a,b])\quad,\qquad
X(t):=P_{\gamma|_{S\times[b,t]}}\circ P_{\gamma|_{S\times[a,b]}}[X_x]\quad(t\in[b,c])
\end{align*}
Then $X(t)$ satisfies $(\gamma^*\nabla)_{\partial_t}X=0$ for every $t\in[a,c]$ and has the initial condition
$X(0)=X_x$. By uniqueness of the solution, we thus conclude that
$X(t)=P_{\gamma|_{S\times[a,t]}}[X_x]$.
\end{proof}

\begin{Lem}
\label{lemParallelTranslationSuperlinear}
$P_{\gamma}$ is even (i.e. parity-preserving), $\mO_S$-superlinear and invertible
such that $(P_{\gamma})^{-1}=P_{\gamma^{-1}}$.
\end{Lem}

\begin{proof}
This is shown by standard ODE arguments as follows. Parallel transport is even since the matrix
$\mB(t)$ in (\ref{eqnParallel}) is even. From the same equation, $\mO_S$-linearity is clear.
It is invertible since both $P_{\gamma}$ and $P_{\gamma^{-1}}$ satisfy the same equation
(\ref{eqnParallel}) at $t$ and $1-t$, respectively.
\end{proof}

By restriction, an $S$-connection $\nabla$ on $\mE_S$ induces a connection
$\nabla^{\mE}:\mE\rightarrow\mS M^*\otimes_{\mO_M}\mE$. By further restriction, we obtain
a classical connection $\nabla^0:\Gamma(E)\rightarrow\Gamma(TM_0)\otimes\Gamma(E)$
on the vector bundle $E:=\bigcup_{x\in M_0}\mE_x\rightarrow M_0$ (denoted $\tilde{\nabla}$
in \cite{Gal09}). Let $P_{\gamma_0}:E_{\gamma_0(0)}\rightarrow E_{\gamma_0(1)}$ denote
parallel transport along a path $\gamma_0:[0,1]\rightarrow M_0$ (denoted $\tau_{\gamma}$
in \cite{Gal09}). On the other hand, let $(P_{\gamma})^0:E_{\gamma_0(0)}\rightarrow E_{\gamma_0(1)}$
denote the restriction of $\nabla$-parallel transport along $\gamma:S\times[0,1]\rightarrow M$.

\begin{Lem}
\label{lemUnderlyingParallelTransport}
Let $\gamma:x\rightarrow y$ be an $S$-path. Then $(P_{\gamma})^0=P_{\gamma_0}$.
\end{Lem}

\begin{proof}
This follows immediately from (\ref{eqnParallel}). Note that $\partial_t(\gamma^*(\xi^l))$
is odd, for $\xi^l$ an odd coordinate, and thus projected to zero, leaving only even
indices $l$ in $\gamma^*(\Gamma_{lk}^m)$.
\end{proof}

By Lem. \ref{lemParallelTranslationSuperlinear}, $P_{\gamma}$ is an isomorphism from
$x^*\mE$ to $y^*\mE$. With respect to local bases $(T^k)$ and $(\tilde{T}^k)$
of $\mE$ around $\gamma_0(0)$ and $\gamma_0(1)$, respectively, it can thus be identified
with a matrix in $GL_{\rk\mE}(\mO_S)$.

\begin{Lem}
\label{lemParallelSolution}
The solution to (\ref{eqnParallel}) is given by
\begin{align*}
X(t)=\mP\exp\left(-\int_0^t\mB(\tau)d\tau\right)[X_x]
:=\sum_{j=0}^{\infty}(-1)^j\int_0^td\tau_j\ldots\int_0^{\tau_2}d\tau_1\mB(\tau_j)\cdot\ldots\mB(\tau_1)X_x
\end{align*}
where $X_x\in x^*\mE$ and $x^{\sharp}=\ev|_{t=0}\gamma^{\sharp}$.
\end{Lem}

\begin{proof}
By assumption, $\gamma_0$ takes values in $U_0\subseteq M_0$ such that both $M|_{U_0}$ and $\mE|_{U_0}$ are trivial.
We may thus identify (as vector spaces), for every $t\in[0,1]$, $\ev|_t\gamma^*\mE$ with $\bR^{\rk\mE}\otimes\bigwedge\bR^L\cong\bR^M$
for some $M\in\bN$. With this identification, the $\bR$-linear operator $\mB(t)$ becomes a matrix in $\Mat(M\times M,\bR)$,
and $\partial_tX(t)=-\mB(t)\cdot X(t)$ can be considered as a classical first order linear ordinary differential equation.
It remains to show that the series stated converges absolutely in the Banach space $C^1([0,1],\Mat(M\times M,\bR))$.
Then, differentiating termwise, it follows that it is indeed the solution operator.
These steps are standard. See Lem. 2.6.7 of \cite{Bae09} for a similar treatment.
\end{proof}

\begin{Rem}
Redefining (\ref{eqnLocalConnection}) as
$\Gamma_{\xi^i}[T^j]:=\frac{i}{g}\nabla_{\partial_{\xi^i}}T^j$, we get the
parallelness equation $\partial_tX(t)=ig\mB(t)\cdot X(t)$, and thus the solution operator
\begin{align*}
X(t)=\mP\exp\left(ig\int_0^t\mB(\tau)d\tau\right)[X_x]:=\sum_{j=0}^{\infty}(ig)^j\int_0^td\tau_j\ldots\int_0^{\tau_2}d\tau_1\mB(\tau_j)\cdot\ldots\mB(\tau_1)X_x
\end{align*}
as in (\ref{eqnWilsonLine}). This convention is more usual in the physical literature.
\end{Rem}

An important property of the Wilson loop is its gauge-invariance. We close this chapter
showing that the trace of parallel transport around an $S$-loop is gauge-invariant,
thus qualifying as a model for the super Wilson loop. We restrict attention to local
gauge transformations in a coordinate chart $U\subseteq M$, which is sufficient for
the situation $M\cong\bR^{n|m}$ considered in \cite{BKS12} and avoids the theory
of super principal bundles.

\begin{Def}
\label{defGaugeTrafo}
A (local) gauge transformation is a morphism of supermanifolds
\begin{align*}
V:S\times U\rightarrow GL_{\rk\mE}\quad\textrm{identified with}\qquad
\left(V^{\sharp}(\zeta^{kl})\right)_{kl}\in GL_{\rk\mE}(\mO_{S\times M}(U))
\end{align*}
where $\zeta^{kl}$ denote the global standard coordinates of the super Lie group $GL_{\rk\mE}$.
It acts on sections $\psi\in\mE_S(U)$ and connections $\nabla$ via
\begin{align*}
\psi\mapsto V\cdot\psi\;,\qquad
\Gamma_{\xi^i}\mapsto V\cdot\Gamma_{\xi^i}\cdot V^{-1}-(\partial_{\xi^i}V)V^{-1}
\end{align*}
where $\Gamma_{\xi^i}$ is as in (\ref{eqnLocalConnection}).
\end{Def}

Consider an $S$-path $\gamma:S\times[0,1]\rightarrow U$
and the concatenation $V_{\gamma}:=V\circ\hat{\gamma}:S\times[0,1]\rightarrow GL_{\rk\mE}$. Let $\mB(t)$ be as in (\ref{eqnParallel}) with respect to the original connection $\nabla$ (and $\gamma$)
and $\tilde{\mB}(t)$ be its gauge transformed counterpart. Then
\begin{align*}
\tilde{\mB}(t)&=\partial_t(\hat{\gamma}^*(\xi^l))\cdot\hat{\gamma}^*\left(V\Gamma_{\xi^l}V^{-1}-(\partial_{\xi^l}V)V^{-1}\right)\\
&=V_{\gamma}\cdot\mB(t)\cdot V_{\gamma}^{-1}-(\partial_tV_{\gamma})\cdot V_{\gamma}^{-1}
\end{align*}
It follows that
\begin{align*}
(\gamma^*\tilde{\nabla})_{\partial_t}(V_{\gamma}\cdot X)=\left(\partial_t+V_{\gamma}\cdot\mB(t)\cdot V_{\gamma}^{-1}-(\partial_tV_{\gamma})\cdot V_{\gamma}^{-1}\right)V_{\gamma}\cdot X
=V_{\gamma}\cdot(\gamma^*\nabla)_{\partial_t}X
\end{align*}
In particular, $X\in\gamma^*\mE$ is $\nabla$-parallel if and only if
$V_{\gamma}\cdot X\in\gamma^*\mE$ is $\tilde{\nabla}$-parallel.

Now let $X_x\in x^*\mE$, and let $\gamma:x\rightarrow y$ connect the $S$-points $x$ and $y$.
Then, $V_x\cdot X_x$ with $V_x:=V\circ\hat{x}$ is moved by $\tilde{\nabla}$-parallel transport
to $V_y$ times $\nabla$-parallel transport of $X_x$. We thus arrive at the following result.

\begin{Prp}
Let $\tilde{P}$ denote parallel transport with respect to the gauge transformed connection $\tilde{\nabla}$. Then
$\tilde{P}=V_y\cdot P\cdot V_x^{-1}$. In particular, if $\gamma:x\rightarrow x$ is closed,
\begin{align*}
\tilde{P}=V_x\cdot P\cdot V_x^{-1}
\end{align*}
and the trace $\tr P=\tr\tilde{P}$ is a gauge invariant quantity.
\end{Prp}

By now, we have achieved the first aim of this article of constructing a mathematical model
of super Wilson loops. Superpoints are $S$-points, and a super Wilson loop is the gauge-invariant
trace of parallel transport around an $S$-loop. The exact choice of $S=\bR^{0|L}$ is not
important, except that $L$ should be sufficiently large to make calculations consistent.
By means of $S$, the super Wilson loop acquires an (unphysical) inner structure.

\section{The Holonomy of an $S$-Point}
\label{secHolonomy}

Let $\mE$ continue to denote a super vector bundle over a supermanifold $M$ and $\nabla$
be an $S$-connection on $\mE_S$ with $S$ a superpoint (\ref{eqnSuperpoint}). In this section,
we define the holonomy group of an $S$-point $x:S\rightarrow M$ and prove an analogon
of the Ambrose-Singer theorem. After endowing the holonomy group to a functor,
we establish a holonomy principle in this context, whose proof makes use of at least
$(\dim M)_{\overline{1}}$ additional Graßmann generators.

\begin{Def}
A piecewise smooth $S$-homotopy is a map
$\Xi:S\times[0,1]\setminus\{t_0,\ldots,t_l\}\times[0,1]\rightarrow M$ such that,
denoting the real coordinates by $t$ and $s$, respectively,
\begin{enumerate}
\renewcommand{\labelenumi}{(\roman{enumi})}
\item the prescription $\Xi_{s_0}^{\sharp}:=\ev|_{s=s_0}\Xi^{\sharp}$ yields
a piecewise smooth $S$-path $\Xi_{s_0}$ for every $s_0\in[0,1]$, and
\item $\Xi^{\sharp}(f)$ is smooth in $s$ for every $f\in\mO_M$.
\end{enumerate}
$\Xi$ is called proper if $\ev_{s,t=0}\Xi^{\sharp}=p^{\sharp}$ and
$\ev_{s,t=1}\Xi^{\sharp}=q^{\sharp}$ for all $s\in[0,1]$ and $S$-points $p$ and $q$.
\end{Def}

\begin{Def}[$S$-Holonomy]
\label{defHolonomy}
Let $x:S\rightarrow M$ be an $S$-point. We set
\begin{align*}
&\Hol_x:=\{P_{\gamma}\setsep\gamma:x\rightarrow x\;\textrm{piecewise smooth}\}\subseteq\End_{\mO_S}(x^*\mE)\\
&\Hol^0_x:=\{P_{\gamma}\setsep\gamma:x\rightarrow x\;\textrm{piecewise smooth and contractible}\}
\end{align*}
with contractible in the sense that
there exists a piecewise smooth proper homotopy $\Xi$ such that $\Xi_0=x$ and $\Xi_1=\gamma$.
\end{Def}

By Lem. \ref{lemParallelTranslationSuperlinear}, $\Hol_x$ is a group which can be identified
with a subgroup of $GL_{\rk\mE}(\mO_S)$ with respect to a local basis $(T^k)$ of $\mE$.
By Thm. \ref{thmAmbroseSinger} below, it is indeed a Lie group.
For $S=\bR^{0|0}$, it follows by Lem. \ref{lemUnderlyingParallelTransport} that
$\Hol_x=\Hol_{\nabla^0}(x_0(0))$ is the holonomy group with respect to the underlying
connection $\nabla^0$.

We call $M$ path-connected if, for any two $S$-points $x$, $y$, there is an $S$-path
$\gamma:x\rightarrow y$. By the following result this, as well as contractability,
is determined by the classical counterparts such that, in particular, $\Hol_x$ does
not depend on the restriction of $M$ to any connected component of $M_0$ different
from that of $x_0(0)$.

\begin{Lem}
\label{lemContractible}
$M$ is path-connected if and only if $M_0$ is. Moreover, a piecewise smooth $S$-loop
$\gamma:x\rightarrow x$ is contractible to $x$ if and only if $\gamma_0$ is contractible to $x_0$.
\end{Lem}

\begin{proof}
It is clear that path-connectedness of $M$ implies that of $M_0$. Conversely, let $x,y:S\rightarrow M$ and $\gamma_0:x_0(0)\rightarrow y_0(0)$ be a connecting classical path. Let
$t_j\in[0,1]$ be such that $\gamma_0|_{[t_j,t_{j+1}]}$
is smooth and its image is contained in the open set $U_0$ for a coordinate chart
$U\subseteq M$ with coordinates $(\xi^k)$.
Any (smooth) morphism $\gamma^j:S\times[t_j,t_{j+1}]\rightarrow U$
can be identified with $(\dim M)_{\overline{0}}+(\dim M)_{\overline{1}}$
smooth maps $\gamma^{j\sharp}(\xi^k):[t_j,t_{j+1}]\rightarrow\bigwedge\bR^L$ or, equivalently,
with a single smooth map $\tilde{\gamma}^j:[t_j,t_{j+1}]\rightarrow\bR^M$ for some $M\in\bN$.
An $S$-path $\gamma:x\rightarrow y$ with underlying path $\gamma_0$ can then
be constructed by glueing together suitable maps $\tilde{\gamma}^j$.
The details are standard and thus omitted.
The proof of the second statement is similar.
\end{proof}

\subsection{An Ambrose-Singer Theorem}

The classical Ambrose-Singer theorem characterises the holonomy Lie algebra in terms
of the curvature of the connection considered. In this section, we show that this theorem
continues to hold in the more general situation of $S$-holonomy in the sense of
Def. \ref{defHolonomy}. Our proof is modelled on a classical proof due to Levi-Civita
as presented in \cite{Bal02}. We define the curvature of $\nabla$ as usual by
\begin{align}
\label{eqnCurvature}
\scal[R]{X}{Y}Z:=\nabla_X\nabla_YZ-(-1)^{\abs{X}\abs{Y}}\nabla_Y\nabla_XZ-\nabla_{\scal[[]{X}{Y}}Z
\end{align}
for $X,Y\in\mS M_S$ and $Z\in\mE_S$, where $\scal[[]{X}{Y}:=XY-(-1)^{\abs{X}\abs{Y}}YX$
is the supercommutator. This definition is such that
$R\in\scal[\Hom_{\mO_{S\times M}}]{\mS M_S\otimes_{\mO_{S\times M}}\mS M_S\otimes_{\mO_{S\times M}}\mE_S}{\mE_S}_{\overline{0}}$. The curvature is skew-symmetric
\begin{align*}
\scal[R]{Y}{X}=-(-1)^{\abs{X}\abs{Y}}\scal[R]{X}{Y}
\end{align*}
which is inherited to the pullback. Let $\varphi:S\times N\rightarrow M$ be a supermanifold morphism. Then
\begin{align}
\label{eqnRSkewSymmetric}
\scal[R_{\varphi}]{A}{B}=-(-1)^{\abs{A}\abs{B}}\scal[R_{\varphi}]{B}{A}
\end{align}
for $A,B\in\varphi^*\mS M$.
This is shown by a straightforward calculation in coordinates, writing $A=(\varphi^*\partial_{\xi^k})\cdot A^k$ etc.

\begin{Def}
\label{defHolAlgGL}
Let $x:S\rightarrow M$ be an $S$-point. Let $\fg_x$ denote the Lie subalgebra of
$(\fgl_{\rk\mE}(\bigwedge\bR^L))_{\overline{0}}$ which is generated by the following set of endomorphisms.
\begin{align*}
\{P_{\gamma}^{-1}\circ\scal[R_y]{u}{v}\circ P_{\gamma}\setsep
y:S\rightarrow M\,,\;\gamma:x\rightarrow y\;\textrm{piecewise smooth}\,,\;u,v\in(y^*\mS M)_{\overline{0}}\}
\end{align*}
\end{Def}

$\Hol_x$ is contained in $GL_{\rk\mE}(\bigwedge\bR^L)$. By the
following lemma, this is a Lie group.
In general, every Lie subalgebra of the Lie algebra of a Lie group is the Lie algebra of a
unique immersed connected Lie subgroup (see Chp. 2 of \cite{GOV97}).
Let $G_x\subseteq GL_{\rk\mE}(\bigwedge\bR^L)$ denote this Lie subgroup corresponding to
$\fg_x\subseteq(\fgl_{\rk\mE}(\bigwedge\bR^L))_{\overline{0}}$.

\begin{Lem}
$GL_{n|m}(\bigwedge\bR^L)$ is a real Lie group with Lie algebra
$(\fgl_{n|m}(\bigwedge\bR^L))_{\overline{0}}$.
\end{Lem}

\begin{proof}
$M\in(\fgl_{n|m}(\bigwedge\bR^L))_0$ is invertible if and only if its image under the canonical
projection to $\fgl_{n|m}(\bR)$ is (Lem. 3.6.1 in \cite{Var04}). Therefore
\begin{align*}
GL_{n|m}\left(\bigwedge\bR^L\right)=\left(GL_n(\bR)\times GL_m(\bR)\right)\oplus\left(\fgl_{n|m}\left(\bigwedge(\bR^L)_{\mathrm{nilpotent}}\right)\right)_{\overline{0}}
\end{align*}
which is open in $(\fgl_{n|m}(\bigwedge\bR^L))_{\overline{0}}$ and as such
a submanifold with a group structure such that the tangent space at $1$ can be identified
with $(\fgl_{n|m}(\bigwedge\bR^L))_{\overline{0}}$.
Writing the matrix entries of a product $M\cdot L$ in terms of real coefficients of odd
generators, it is clear that multiplication is smooth, and similar for inversion. One
further shows that the Lie algebra commutator coincides with the commutator
$\scal[[]{X}{Y}=XY-YX$.
\end{proof}

\begin{Thm}[Ambrose-Singer Theorem]
\label{thmAmbroseSinger}
The Lie groups $G_x=\Hol^0_x$ coincide. In particular, $\Hol_x$ is a Lie group with identity
component $\Hol_x^0$ and Lie algebra $\hol_x=\fg_x$.
\end{Thm}

We defer the proof of the theorem to the end of the present section. It is based on Prp.
\ref{prpParallelDifferential} and Prp. \ref{prpDerivativeOfParallelTranslation} below.
The following two lemmas are needed in the proof of the first proposition.

\begin{Lem}
\label{lemPullbackCurvature}
Let $f:S\times[a,b]\times[b,c]\rightarrow M$ be a morphism and $X\in f^*\mE$ be a section along $f$. Then
\begin{align*}
(f^*\nabla)_{\partial_s}(f^*\nabla)_{\partial_t}X-(f^*\nabla)_{\partial_t}(f^*\nabla)_{\partial_s}X
=\scal[R_f]{df[\partial_s]}{df[\partial_t]}X
\end{align*}
where $(s,t)$ denote the standard coordinates on $[a,b]\times[b,c]$.
\end{Lem}

\begin{proof}
This is shown by a direct calculation in local coordinates $(\xi^k)$ of $M$ and a trivialisation $(T^k)$ of $\mE$, writing $X=(\varphi^*T^l)\cdot X^l$ with $X^l\in\mO(S\times[a,b]\times[b,c])$.
\end{proof}

Let $x,y:S\rightarrow M$ and $\gamma:x\rightarrow y$.
A tuple $(e^1,\ldots,e^k)$ of sections $e^j\in\gamma^*\mE$ is a basis of $\gamma^*\mE$ if and only if
$(\ev|_{t=t_0}e^1,\ldots,\ev|_{t=t_0}e^k)$ is a basis of $\ev|_{t=t_0}\gamma^*\mE$ for every $t_0\in[0,1]$.
It is called parallel if all $e^i$ are parallel. Such a basis is determined by its evaluation at $t=0$
via $\ev|_{t=t_0}e^j=P_{\gamma|_{S\times[0,t_0]}}(\ev|_{t=0}e^j)$. In particular,
a parallel basis, as used in the proof of the following lemma, exists.

\begin{Lem}
\label{lemPtDt}
Let $X\in\gamma^*\mE$ be a section along $\gamma$. Let $P_t:=P_{\gamma|_{S\times[0,t]}}^{-1}$
be the parallel dispacement from $\ev|_t\gamma^{\sharp}$ to $x^{\sharp}=\ev|_{t=0}\gamma^{\sharp}$. Then
\begin{align*}
P_t\ev|_t(\gamma^*\nabla)_{\partial_t}X=\partial_t P_t(\ev|_tX)\in x^*\mE
\end{align*}
\end{Lem}

\begin{proof}
Let $(e^j)$ be a parallel basis along $\gamma$.
Writing $X=e^i\cdot X^i$ with $X^i\in\mO_{S\times[0,1]}$, it follows that
$P_t(\ev|_tX)=\ev_{t=0}e^i\cdot\ev|_tX^i$, and
\begin{align*}
\partial_t P_t(\ev|_tX)=\ev|_{t=0}e^i\cdot\ev|_t(\partial_tX^i)
\end{align*}
On the other hand, $(\gamma^*\nabla)_{\partial_t}X=e^i\cdot\partial_t(X^i)$ implies
\begin{align*}
P_t\ev|_t(\gamma^*\nabla)_{\partial_t}X=P_t(\ev|_te^i\cdot\ev|_t(\partial_tX^i))
=\ev|_{t=0}e^i\cdot\ev|_t(\partial_tX^i)
\end{align*}
such that both sides agree.
\end{proof}

For the following proposition note that, for a proper $S$-homotopy $\Xi$,
we may identify $\ev|_{t=0}\Xi^{\sharp}$ and $\ev|_{t=1}\Xi^{\sharp}$
with single $S$-points $x,y:S\rightarrow M$, respectively.

\begin{Prp}
\label{prpParallelDifferential}
Let $\Xi$ be a proper $S$-homotopy, and let
$P_{s,t}:=P_{\Xi_s|_{S\times[t,1]}}$ denote parallel transport along the restriction of
the $S$-path $\Xi_s$ to $S\times[t,1]$. Then
\begin{align*}
\partial_sP_{s,0}=\left(\int_0^1R_{s,t}dt\right)P_{s,0}
\in\scal[\Hom_{\mO_{S\times[0,1]}}]{x^*\mE}{y^*\mE}
\end{align*}
with $R_{s,t}:=P_{s,t}\ev|_{s,t}\scal[R_{\Xi}]{d\Xi[\partial_t]}{d\Xi[\partial_s]}P_{s,t}^{-1}$
\end{Prp}

\begin{proof}
Let $Z\in\Xi^*\mE$. For $\Xi$ proper, the term $\partial_s(\Xi_s^*(\xi^l))$ in (\ref{eqnParallel})
vanishes for $t=0$ as well as $t=1$, such that
\begin{align*}
\ev|_{s,t=0}(\Xi^*\nabla)_{\partial_s}Z=\partial_s\ev|_{s,t=0}Z\;,\qquad
\ev|_{s,t=1}(\Xi^*\nabla)_{\partial_s}Z=\partial_s\ev|_{s,t=1}Z
\end{align*}
Consider $Z$ such that the first term vanishes and, moreover,
$(\Xi^*\nabla)_{\partial_t}Z\equiv 0$. By Lem. \ref{lemPtDt} and Lem. \ref{lemPullbackCurvature},
we yield
\begin{align*}
\partial_tP_{s,t}\ev|_{s,t}(\Xi^*\nabla)_{\partial_s}Z
&=P_{s,t}\ev_{s,t}(\Xi^*\nabla)_{\partial_t}(\Xi^*\nabla)_{\partial_s}Z\\
&=P_{s,t}\ev|_{s,t}\scal[R_{\Xi}]{d\Xi[\partial_t]}{d\Xi[\partial_s]}Z\\
&=R_{s,t}\ev_{s,t=1}Z
\end{align*}
since $\ev_{s,t}Z=P_{s,t}^{-1}\ev_{s,t=1}Z$ by assumption. This, together with the assumptions
on $Z$ and $P_{s,1}=\id$, implies the following.
\begin{align*}
\partial_sP_{s,0}\ev|_{s,t=0}Z&=\partial_s\ev|_{s,t=1}Z\\
&=P_{s,1}\ev|_{s,t=1}(\Xi^*\nabla)_{\partial_s}Z-P_{s,0}\ev|_{s,t=0}(\Xi^*\nabla)_{\partial_s}Z\\
&=\int_0^1\partial_t\left(P_{s,t}\ev_{s,t}(\Xi^*\nabla)_{\partial_s}Z\right)dt\\
&=\left(\int_0^1R_{s,t}dt\right)P_{s,0}\ev|_{s,t=0}Z
\end{align*}
Let $Z_x\in x^*\mE$ .Then, setting $Z(s,t):=P_{\Xi_s|_{S\times[0,t]}}Z_x$, defines
a section $Z\in\Xi^*\mE$ that satisfies the assumptions made in the beginning of the proof
as well as $\ev|_{s,t=0}Z=Z_x$, such that the equation to be proved holds applied to $Z_x$.
Since $Z_x$ was arbitrary, it holds in general.
\end{proof}

Let $a:S\rightarrow M$ be an $S$-point and $u,v\in(a^*\mS M)_{\overline{0}}$.
With respect to local coordinates
$(\xi^i)$ on $U\subseteq M$ around $a_0(0)$, we write $u=(a^*\partial_{\xi^i})\cdot u^i$
with $u^i\in\mO_S$ and likewise for $v$. Let $(x,y)$ denote standard coordinates of $\bR^2$.
Then the map
\begin{align*}
f:S\times\bR^2\rightarrow U\;,\qquad f^{\sharp}(\xi^i):=a^{\sharp}(\xi^i)+(-1)^{\abs{\xi^i}}u^i\cdot x+(-1)^{\abs{\xi^i}}v^i\cdot y
\end{align*}
is such that
\begin{align}
\label{eqnF}
\ev|_{(x,y)=(0,0)}f^{\sharp}=a^{\sharp}\;,\qquad\ev|_{(0,0)}df[\partial_x]=u\;,\qquad
\ev|_{(0,0)}df[\partial_y]=v
\end{align}
Consider also the following piecewise smooth homotopy $g:S\times[0,1]\times[0,1]\rightarrow\bR^2$.
\begin{align*}
g_0(s,t):=\left\{\arr{ll}{(4st,0)&0\leq t\leq 1/4\\(s,s(4t-1))&1/4\leq t\leq 1/2\\(s(3-4t),s)&1/2\leq t\leq 3/4\\(0,4s(1-t))&3/4\leq t\leq 1}\right.\;,\qquad
\arr{c}{g^{\sharp}(x):=g_0^*(x)\\g^{\sharp}(y):=g_0^*(y)}
\end{align*}

\begin{Prp}
\label{prpDerivativeOfParallelTranslation}
Let $a:S\rightarrow M$ be an $S$-point and $u,v\in(a^*\mS M)_{\overline{0}}$. Let $f$ be such that
(\ref{eqnF}), and let $P_s$ denote parallel translation along $\Xi_s$ for
\begin{align*}
\Xi:=f\circ\hat{g}:S\times[0,1]\times[0,1]\rightarrow U\subseteq M
\end{align*}
Then
\begin{align*}
\ev|_{s=0}\partial_sP_s=0\;,\qquad\ev|_{s=0}\partial_s\partial_sP_s=2\scal[R_a]{v}{u}
\end{align*}
\end{Prp}

\begin{proof}
By Lem. \ref{lemChainRule}, we have $d\Xi[\partial_t]=(\Xi^*\partial_l)g^{\sharp}(\phantrix{df}{l}{i})\partial_tg^{\sharp}(x^i)$ where $x^i$ runs over $x$ and $y$.
For $t\leq 1/4$,
\begin{align*}
\scal[R_{\Xi}]{d\Xi[\partial_t]}{d\Xi[\partial_s]}
=\scal[R_{\Xi}]{(\Xi^*\partial_l)g^{\sharp}(\phantrix{df}{l}{x})4s}{(\Xi^*\partial_m)g^{\sharp}(\phantrix{df}{m}{x})4t}=0
\end{align*}
vanishes by skew-symmetry (\ref{eqnRSkewSymmetric}), and analogous for $t\geq 3/4$.
For $1/4\leq t\leq 3/4$, we find
\begin{align*}
\scal[R_{\Xi}]{d\Xi[\partial_t]}{d\Xi[\partial_s]}
=-\scal[R_{\Xi}]{(\Xi^*\partial_l)g^{\sharp}(\phantrix{df}{l}{x})}{(\Xi^*\partial_m)g^{\sharp}(\phantrix{df}{m}{y})}\cdot 4s
\end{align*}
Using (\ref{eqnF}), we further calculate, for $1/4\leq t\leq 3/4$,
\begin{align*}
R_{s,t}&=P_{s,t}\ev|_{s,t}\scal[R_{\Xi}]{d\Xi[\partial_t]}{d\Xi[\partial_s]}P_{s,t}^{-1}\\
&=4s\cdot P_{s,t}\ev|_{s,t}\scal[R_{\Xi}]{(\Xi^*\partial_l)g^{\sharp}(\phantrix{df}{l}{y})}{(\Xi^*\partial_m)g^{\sharp}(\phantrix{df}{m}{x})}P_{s,t}^{-1}\\
&=4s\cdot P_{s,t}\scal[R_a]{(a^*\partial_l)v^l}{(a^*\partial_m)u^m}P_{s,t}^{-1}
\end{align*}
Prp. \ref{prpParallelDifferential} now yields
\begin{align*}
\partial_sP_s=4s\left(\int_{\frac{1}{4}}^{\frac{3}{4}}P_{s,t}\scal[R_a]{v}{u}P_{s,t}^{-1}dt\right)P_s
\end{align*}
which vanishes for $s\rightarrow 0$. Likewise
\begin{align*}
\ev_{s=0}\partial_s(\partial_sP_s)=\lim_{s\rightarrow 0}\left(\frac{1}{s}4s\left(\int_{\frac{1}{4}}^{\frac{3}{4}}P_{s,t}\scal[R_a]{v}{u}P_{s,t}^{-1}dt\right)P_{s}\right)
=2\scal[R_a]{v}{u}
\end{align*}
\end{proof}

\begin{proof}[Proof of Thm. \ref{thmAmbroseSinger}]
Let $\gamma:x\rightarrow x$ be piecewise smooth and contractible.
We choose a piecewise smooth proper homotopy $\Xi$ such that
$\Xi_0=x$ and $\Xi_1=\gamma$, and let
$P_s:=P_{\Xi_s}\in GL_{\rk\mE}(\bigwedge\bR^L)$ denote parallel translation
along $\Xi_s$. By Prp. \ref{prpParallelDifferential}, it satisfies the differential equation
\begin{align*}
\partial_sP_s=g(s)\cdot P_s\;,\qquad g(s):=\left(\int_a^bR_{s,t}dt\right)\in\fg_x
\end{align*}
By standard Lie group theory (cf. Chp. 2 of \cite{GOV97}), we conclude that
$P_s\in G_x$ and, in particular, $P_{\gamma}=P_1\in G_x$. Therefore,
$\Hol^0_x\subseteq G_x$ is a path-connected subgroup. By a theorem of Yamabe \cite{Yam50},
it is a Lie subgroup.

Let $a$ be an $S$-point, $\gamma:x\rightarrow a$ and $u,v\in(a^*\mS M)_{\overline{0}}$.
Let $\Xi$ be as in Prp. \ref{prpDerivativeOfParallelTranslation}, and let
$P_s\in\Hol^0_x$ denote parallel translation along
$\hat{\Xi}_s:=\gamma\star\Xi_s\star\gamma^{-1}$. Then
\begin{align*}
\partial_sP_s|_{s=0}&=P_{\gamma}\circ\partial_s P_{\Xi_s}|_{s=0}\circ P_{\gamma}^{-1}=0\\
\partial_s\partial_sP_s|_0&=P_{\gamma}\circ\partial_s\partial_s P_{\Xi_s}|_0\circ P_{\gamma}^{-1}
=2 P_{\gamma}\circ\scal[R]{v}{u}\circ P_{\gamma}^{-1}
\end{align*}
by Prp. \ref{prpDerivativeOfParallelTranslation}.
$\Hol^0_x$ can be identified with a submanifold of some $\bR^M$. By the vanishing
of the first derivative we can thus conclude that
$\partial_s\partial_s P_s|_0\in\hol_x=T_e(\Hol_x^0)$.
Therefore, all generators of $\fg_x$ are contained in $\hol_x$.
It follows that $\fg_x=\hol_x$ and $\Hol^0_x=G_x$.
\end{proof}

\subsection{The Holonomy Group Functor}

So far, we have considered a fixed superpoint $S=\bR^{0|L}$ along with an $S$-connection
$\nabla$ on an $S$-bundle $\mE_S$. In Sec. \ref{secSuperWilsonLoops}, it was argued that having
$S$-connections (compared to plain connections in $\mE$) is necessary to model superconnections
as in \cite{BKS12}, whereas the exact value of $L$ cannot have any physical significance.
But also for purely mathematical reasons, it is desirable to allow for extending the number of
auxiliary Graßmann generators, as will become clear in the proof of the holonomy principle
(Thm. \ref{thmHolonomyPrinciple}) below. This extension results in a categorical theory
to be described next.

Let $\nabla$ be an $S$-connection on $\mE_S$ with respect to $S=\bR^{0|L}$, and let
$T=\bR^{0|L'}$ be another superpoint. By $\bigwedge\bR^{L'}$-linear extension, $\nabla$
can be considered as an $S\times T$-connection on $\mE_{S\times T}$. Similarly,
an $S$-point $x:S\rightarrow M$ canonically induces an $S\times T$-point $x_T:S\times T\rightarrow M$
by composing $x$ with the canonical projection $S\times T\rightarrow S$.
For the next proposition, note that a morphism $\varphi:T\rightarrow T'$ can be identified
with a Graßmann algebra morphism $\varphi^*$ and as such acts naturally on $GL_{\rk\mE}(\mO_{T'})$.

\begin{Prp}
\label{prpHolonomyFunctor}
The assignment
\begin{align*}
T\mapsto\Hol_x(T):=\Hol_{x_T}\;,\qquad
(\varphi:T\rightarrow T')\mapsto\left(L\mapsto\varphi^*(L),\,\Hol_{x_{T'}}\rightarrow\Hol_{x_T}\right)
\end{align*}
defines a group-valued functor.
\end{Prp}

In the following, we will denote both the holonomy with respect to $x$ and the induced
holonomy functor by $\Hol_x$. We will also use the notation $\hol_x(T):=\hol_{x_T}$.

\begin{proof}
Let $L\in\Hol_{x_{T'}}$. We must show that the pullback $\varphi^*(L)$ is
indeed contained in $\Hol_{x_T}$. Then the induced map $\Hol_{x_{T'}}\rightarrow\Hol_{x_T}$ is clearly a group homomorphism.

Let $\gamma:x_{T'}\rightarrow x_{T'}$ be such that $L=P_{\gamma}$, and prescribe
\begin{align*}
x_{\varphi}:=x_{T'}\circ(\id_S\times\varphi):S\times T\rightarrow M\;,\quad
\gamma_{\varphi}:=\gamma\circ\varphi:=\gamma\circ(\id_S\times\varphi\times\id_{[0,1]}):x_{\varphi}\rightarrow x_{\varphi}
\end{align*}
It is clear that $x_{\varphi}=x_T$ independent of $\varphi$. Let $\mB(t)$ be as in (\ref{eqnParallel}) with respect to $\gamma$. It follows that the local parallelness condition with respect to
$\gamma_{\varphi}$ reads
\begin{align*}
\partial_tX(t)=-\left(\varphi^*\mB(t)\right)\cdot X(t)
\end{align*}
We can, therefore, conclude that $X\in\gamma^*\mE$ parallel along $\gamma$ implies that
$\varphi^*X\in\gamma_{\varphi}^*\mE$ is parallel along $\gamma_{\varphi}$. Therefore
\begin{align*}
\varphi^*\left(P_{\gamma}[X_{x_{T'}}]\right)=P_{\gamma_{\varphi}}\left[\varphi^*(X_{x_{T'}})\right]\quad\textrm{for all}\;\;X_{x_{T'}}\in x_{T'}^*\mE
\end{align*}
and $\varphi^*(L)=\varphi^*P_{\gamma}=P_{\gamma_{\varphi}}\in\Hol_{x_T}$.
\end{proof}

The Molotkov-Sachse theory defines a supermanifold to be a certain functor from the category
$\Gr$ of Graßmann algebras to that of smooth manifolds (\cite{Mol84,Sac08})
such that, in the finite-dimensional case, the resulting category is equivalent to that
of Berezin-Kostant-Leites supermanifolds.
It is thus natural to conjecture that $\Hol_x$ is representable in that it defines
such a supermanifold. If this was true, a neighbourhood of $1$ in $\Hol_x(T)$
would be isomorphic to $(V\otimes\bigwedge\bR^{L'})_{\overline{0}}$
for a fixed finite-dimensional super vector space $V$. It would follow that
$\hol_x(T)\cong T_e(\Hol_x(T))\cong(V\otimes\bigwedge\bR^{L'})_{\overline{0}}$
such that, in particular, $\hol_x(\bigwedge\bR^0)=V_{\overline{0}}$.
The following example shows that the holonomy functor is, in general, not representable.

\begin{Exp}
\label{expGalaev}
Consider $S:=\bR^{0|0}$ and $M:=\bR^{0|1}$ with the ($S$-)connection defined by
$\nabla_{\partial_{\theta}}\partial_{\theta}=\theta\partial_{\theta}$ on
$\mE_S:=\mS M_S=\mS M$ such that $\scal[R]{\partial_{\theta}}{\partial_{\theta}}\partial_{\theta}=2\partial_{\theta}$. Let $0$ denote the unique $S$-point corresponding to $0\in\bR^0$.
By Thm. \ref{thmAmbroseSinger}, $\hol_0(T)$ is generated by
$P_{\gamma}^{-1}\circ\scal[R_y]{u}{v}\circ P_{\gamma}$ for $y:T\rightarrow M$,
$\gamma:x\rightarrow y$ and $u,v\in(y^*\mS M)_{\overline{0}}$.
We write $u=(y^*\partial_{\theta})\cdot u^{\theta}$ with $u^{\theta}\in(\mO_T)_{\overline{1}}$ and analogous for $v$. Let $w\in y^*\mS M$. Then a short calculation yields
\begin{align*}
P_{\gamma}^{-1}\circ\scal[R_y]{u}{v}P_{\gamma}[w]
=-2u^{\theta}v^{\theta}\cdot w
\end{align*}
For $T=\bR^{0|0}$, $u^{\theta}$ and $v^{\theta}$ vanish, such that $\hol_0=\{0\}$ is trivial,
while $\hol_0(T)=\fgl(0|1)\otimes\left((\mO_T)_{\overline{1}}\right)^2\subseteq\fgl(0|1)\otimes(\mO_T)_{\overline{0}}=\left(\fgl(0|1)\otimes\mO_T\right)_{\overline{0}}$ for $T=\bR^{0|L'}$, $L'\geq 2$. By the preceding paragraph, the functor $\Hol_0(T)$ is thus not representable.
\end{Exp}

By the holonomy principle, to be established next, a parallel section $X\in\mE_S$ is uniquely
determined by its $\Hol_x(T)$-invariant pullback $x^*X\in x^*\mE$ as defined in
(\ref{eqnSectionPullback}), where the number $L'$ of additional generators must be sufficiently
large.

\begin{Thm}[Holonomy Principle]
\label{thmHolonomyPrinciple}
Let $M$ be connected.
Let $\nabla$ be an $S$-connection on $\mE_S$, $x:S\rightarrow M$ be an $S$-point and
$T=\bR^{0|L'}$ with $L'\geq(\dim M)_{\overline{1}}$. Then the following holds true.
\begin{enumerate}
\renewcommand{\labelenumi}{(\roman{enumi})}
\item Let $X\in\mE_S$ be a parallel section $\nabla X\equiv 0$ and define $X_x:=x^*X\in x^*\mE$.
Then, for all $y:S\times T\rightarrow M$ and
$\gamma:x\rightarrow y$, it holds $y^*X=P_{\gamma}[X_x]$, where $X_x$ is identified with
a section of $x_T^*\mE$. In particular, $X_x$ is holonomy invariant $\Hol_x(T)\cdot X_x=X_x$.
\item Conversely, let $X_x\in x^*\mE$ be a section such that $\Hol_x(T)\cdot X_x=X_x$.
Then there exists a unique section $X\in\mE_S$ with $x^*X=X_x$, which is parallel $\nabla X\equiv 0$.
\end{enumerate}
\end{Thm}

\begin{proof}
Let $\gamma:x\rightarrow y$ be a piecewise smooth $S$-path.
The assumption $\nabla X\equiv 0$ implies $\nabla_{\partial_t}(\gamma^*X)=0$.
Parallel transport along $\gamma$ is thus
\begin{align*}
P_{\gamma}[X_x]=\ev|_{t=1}\gamma^*X=y^*X
\end{align*}
which proves the first assertion.

Conversely, let $X_x\in x^*\mE$ be such that $\Hol_x(T)\cdot X_x=X_x$.
For a superpoint $y:S\times T\rightarrow M$, we define $X_y:=P_{\gamma}[X_x]$ where
$\gamma:x\rightarrow y$ is an $S\times T$-path.
Since $X_x$ is $\Hol_x(T)$-invariant, $X_y$ is well-defined independent of
the choice of $\gamma$.
We aim at constructing $X$ out of the set of $X_y$ inductively over the degree of $\mO_S$-monomials.
Without loss of generality, we may assume that $M\cong\bR^{n|m}$ has global coordinates
$\xi=(x,\theta)$. For, assume that the statement is true for $M$ replaced by a neighbourhood
$U\subseteq M$ of $x_0(0)$, thus resulting in a parallel section $X\in\mE_S(U)$. Then, by the
first part of the theorem, $X$ satisfies $\Hol_y(T)\cdot X_y=X_y$ for all $y:S\times T\rightarrow U$.
Repeating the local construction in a neighbourhood $V\subseteq M$ of $y_0(0)$ yields
a parallel section $\tilde{X}\in\mE_S(V)$ which, by uniqueness, agrees with $X$ on the
intersection $U_0\cap V_0$. Without loss of generality, we may further assume that $\mE$
is trivial with a global adapted basis $(T^j)$. We expand
\begin{align*}
X_y=X_y|_{\eta^I}\cdot\eta^I=T^j\cdot X_y^j|_{\eta^I}\cdot\eta^I\;,\qquad
X=X|_{\eta^I}\cdot\eta^I=T^j\cdot X^j|_{\eta^I}\cdot\eta^I
\end{align*}
for multiindices $I=(i_1,\ldots,i_{\abs{I}})$ with $1\leq i_j\leq L$, such that
$X_y^j|_{\eta^I}\in\bR$ and $X^j|_{\eta^I}\in\mO_M$ and $X|_{\eta^I}\in\mE$.
Similarly, $\nabla$ is characterised by
$\Gamma_{ij}^k=\left(\Gamma_{ij}^k\right)|_{\eta^I}\cdot\eta^I$.

In the first step, we construct $X^0\in\mE$. Letting $q:=y_0(0)$, we define its value at $q$ by
$X^0(q):=X_y|_{\eta^0}=(P_{\gamma}[X_x])|_{\eta^0}$. By Lem. \ref{lemUnderlyingParallelTransport},
it arises by classical parallel transport along $\gamma_0$.
It is thus independent of $y$ such that $q=y_0(0)$, and $X^0(q)$ depends smoothly on $q$.
By (16) of \cite{Gal09} applied to the induced connection $\nabla^{\mE}$ on $\mE$, $X^0(q)$
extends to a section $X^0\in\mE$ such that $0=\nabla^{\mE}_{\partial_{\theta^j}}X^0
=(\nabla_{\partial_{\theta^j}}X^0)|_{\eta^0}$. By construction, $X^0$ satisfies
$(y^*X^0)|_{\eta^0}=X^0(q)=(P_{\gamma}[X_x])|_{\eta^0}$.
Again by Lem. \ref{lemUnderlyingParallelTransport}, we further note that
$(\nabla X^0)|_{\theta^0\eta^0}\equiv 0$.

In the second step, we consider multiindices $I=(i_1,\ldots,i_{\abs{I}})$ with
$1\leq i_j\leq L+(\dim M)_{\overline{1}}$, such that $\eta^I\in\mO_{S\times T}$.
Assume, by induction, that we have constructed $X^N\in\mE_S$ for $N\in\bN$ such that
\begin{enumerate}
\setcounter{enumi}{-1}
\renewcommand{\labelenumi}{$\arabic{enumi}_N$}
\item $X^N$ has an expansion $X^N=\sum_{\abs{I}\leq N}X|_{\eta^I}\cdot\eta^I$
such that $X|_{\eta^I}=0$ whenever there is $i_j\in I$ with $i_j\geq L+1$.
\item $(y^*X^N)|_{\eta^I}=(P_{\gamma}[X_x])|_{\eta^I}=X_y|_{\eta^I}$
for every $y:S\times T\rightarrow M$, $\gamma:x\rightarrow y$ and $\abs{I}\leq N$.
\item $(\nabla_{\partial_{\theta^j}}X^N)|_{\eta^I}\equiv 0$ for all $\abs{I}\leq N$.
\item $(\nabla_{\partial_{x^j}}X^N)|_{\theta^A\eta^B}\equiv 0$ for all $A$,$B$ such that
$\abs{A}+\abs{B}\leq N$, where $A=(a_1,\ldots,a_{\abs{A}})$ with $1\leq a_j\leq(\dim M)_{\overline{1}}$.
\end{enumerate}
Condition $1_{N+1}$ is equivalent to $1_N$ together with
\begin{align*}
X_y|_{\eta^J}\stackrel{!}{=}(y^*X^{N+1})|_{\eta^J}=y_0^*(X|_{\eta^J})+(y^*X^N)|_{\eta^J}
\qquad\mathrm{for}\;\abs{J}=N+1
\end{align*}
We are thus led to define the value of $X|_{\eta^J}$ at $q$ by
\begin{align}
\label{eqnValueXIk}
X|_{\eta^J}(q):=X_y|_{\eta^J}-(y^*X^N)|_{\eta^J}\qquad\mathrm{for}\;\abs{J}=N+1
\end{align}
This prescription is independent of $y:S\times T\rightarrow M$ such that $y_0(0)=q$. Indeed,
let $y^1$,$y^2$ be two such $S\times T$-points and $\gamma^{1,2}:x\rightarrow y^{1,2}$ be
connecting $S$-paths. Moreover, let $\delta:y^1\rightarrow y^2$ be such that
$\delta_0(t)\equiv q$.
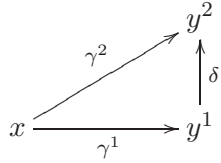
\begin{align*}
\begin{xy}
\xymatrix{
&&y^2\\
x\ar[urr]^{\gamma^2}\ar[rr]_{\gamma^1}&&y^1\ar[u]_{\delta}}
\end{xy}
\end{align*}
Since $X_x$ is holonomy invariant, we have $X_{y^2}=P_{\delta}[X_{y^1}]$.
We calculate, using (\ref{eqnParallel}),
\begin{align*}
\partial_t\left(X_{\delta}|_{\eta^J}-(\delta^*X^N)|_{\eta^J}\right)
&=\left(\partial_tP_{\delta}|_{[0,t]}[X_{y^1}]-\partial_t\delta^*X^N\right)|_{\eta^J}\\
&=\left(-(-1)^{\abs{T^m}(\abs{T^n}+1)}(\delta^*T^m)\partial_t(\delta^*(\xi^l))\cdot\hat{\delta}^*(\Gamma_{ln}^m)\cdot P_{\delta}|_{[0,t]}[X_{y^1}]^n\right.\\
&\left.\qquad\qquad\qquad-\partial_t\delta^*(\xi^l)(\delta^*\circ\partial_{\xi^l})(X^N)\right)|_{\eta^J}
\end{align*}
By assumption, the term $\partial_t(\delta^*(\xi^l))$ is nilpotent such that,
using induction assumption $1_N$, we may replace
$P_{\delta}|_{[0,t]}[X_{y^1}]^n$ by $(\delta^*X^N)^n=\delta^*X^{N\,n}$.
Therefore, the right hand side equals
\begin{align*}
&\left((\delta^*T^m)\partial_t(\delta^*(\xi^l))\hat{\delta}^*\left(-(-1)^{\abs{T^m}(\abs{T^n}+1)}\Gamma_{ln}^mX^{N\,n}-\partial_{\xi^l}X^{N\,m}\right)\right)|_{\eta^J}\\
&\qquad\qquad=-\left((\delta^*T^m)\partial_t(\delta^*(\xi^l))\hat{\delta}^*(\nabla_{\partial_{\xi^l}}X^N)^m\right)|_{\eta^J}
\end{align*}
By $2_N$ and $3_N$ (and nilpotency of $\partial_t(\delta^*(\xi^l))$), this expression
vanishes, thus showing that $X_{\delta}|_{\eta^J}-(\delta^*X^N)|_{\eta^J}$
is constant, which proves that (\ref{eqnValueXIk}) is well-defined.

We next endow $X|_{\eta^J}(q)$ to a section $X|_{\eta^J}\in\mE$ such that
$X^{N+1}:=\sum_{\abs{J}\leq N+1}X|_{\eta^J}\cdot\eta^J$ satisfies $2_{N+1}$.
$2_N$ implies that $(\nabla_{\partial_{\theta^m}}X^{N+1})|_{\eta^I}=0$ with
$\abs{I}\leq N$ for any such $X^{N+1}$. Under this induction hypothesis,
$2_{N+1}$ is thus equivalent to $(\nabla_{\partial_{\theta^m}}X^{N+1})|_{\eta^J}=0$
for $\abs{J}=N+1$ which, in turn, is equivalent to
\begin{align*}
\left(\partial_{\theta^r}\ldots\partial_{\theta^1}\partial_{\theta^m}X^j|_{\eta^J}\right)|_{\theta^0}
=-(-1)^{\abs{T^j}(\abs{T^i}+1)}\left(\partial_{\theta^r}\ldots\partial_{\theta^1}(\Gamma_{mi}^jX^{i\,N+1})\right)|_{\eta^J\theta^0}
\end{align*}
for all $r\leq(\dim M)_{\overline{1}}$. Similar to the construction of $X^0\in\mE$ above,
these equations uniquely determine
$X|_{\eta^J}$, for $\abs{J}=N+1$, by $X^N$ and $X|_{\eta^J}(q)$, such that $2_{N+1}$ holds. If
any index $l_j\in J$ satisfies $l_j>L$, the right hand side of (\ref{eqnValueXIk}) vanishes
upon considering $y:S\rightarrow M$, such that $0_{N+1}$ is satisfied. By construction,
also $1_{N+1}$ holds.

We show that $X^{N+1}$ further satisfies $3_{N+1}$.
$1_{N+1}$ implies that $(z^*X^{N+1})|_{\eta^I}=P_{\delta}[X_y]|_{\eta^I}$ for all $z$
and $\delta:y\rightarrow z$ and $\abs{I}\leq N+1$.
In particular, we let $q\in M_0$ and define $y$ and $\delta$ as follows.
\begin{align*}
y^{\sharp}(x^k):=q^*(x^k)=q^k\,,\;y^{\sharp}(\theta^k):=\eta^{L+k}(\in\mO_T)\;,\quad
\delta^{\sharp}(x^k):=q^k+t\delta^{kk_0}\,,\;\delta^{\sharp}(\theta^k):=\eta^{L+k}
\end{align*}
This is such that $\ev|_{t=0}\delta^{\sharp}=y^{\sharp}$. We thus yield
\begin{align*}
0=\left((\delta^*\nabla)_{\partial_t}(\delta^*X^{N+1})\right)|_{\eta^I}
=\hat{\delta}^*(\nabla_{\partial_{x^{k_0}}}X^{N+1})|_{\eta^I}
\end{align*}
Writing $\nabla_{\partial_{x^{k_0}}}X^{N+1}=:N^{AB}\theta^A\eta^B$ with $\eta^B\in\mO_S$,
we conclude that
\begin{align*}
0=\hat{y}^*(\nabla_{\partial_{x^{k_0}}}X^{N+1})|_{\eta^I}=\left(N^{AB}(q)\cdot\eta^{A_L}\eta^B\right)|_{\eta^I}
\end{align*}
with $A_L$ arising from the multiindex $A$
by shifting all indices by $L$, such that $\eta^{A_L}\in\mO_T$.
For $\abs{A}+\abs{B}=\abs{A_L}+\abs{B}=\abs{I}\leq N+1$, this implies that $N^{AB}(q)=0$.

Proceeding inductively yields a section $X:=X^L=X^{L+(\dim M)_{\overline{1}}}\in\mE_S$
such that the induction hypotheses hold with respect to $L+(\dim M)_{\overline{1}}$.
$X$ is, therefore, parallel. Concerning uniqueness, assume that $\tilde{X}\in\mE_S$ is
a second such section. Then $y^*(X-\tilde{X})=0$ for all $y:S\times T\rightarrow M$
such that $X-\tilde{X}=0$ by an argument analogous to that in the previous proof of $3_{N+1}$.
\end{proof}

\section{Comparison with Galaev's Holonomy Theory}
\label{secComparison}

Considering $S=\bR^{0|0}$,
let $\nabla$ be a connection on a super vector bundle $\mE\rightarrow M$ and $x\in M_0$ be
a (topological) point. In this chapter, we will compare the functor $\Hol_x$
with Galaev's holonomy super Lie group $\Hol_x^{\mathrm{Gal}}$, which
was introduced in \cite{Gal09} by means of a certain Harish-Chandra pair built around
the super Lie algebra $\hol_x^{\mathrm{Gal}}$ generated by endomorphisms
\begin{align*}
P_{\gamma_0}^{-1}\circ\scal[\left(\overline{\nabla}^r_{Y_r,\ldots,Y_1}R\right)_y]{Y}{Z}\circ P_{\gamma_0}:x^*\mE\rightarrow x^*\mE
\end{align*}
with $y\in M_0$, $\gamma_0:x\rightarrow y$, $r\geq 0$ and $Y_1,\ldots,Y_r,Y,Z\in y^*\mS M$,
and where $\overline{\nabla}^r_{Y_r,\ldots,Y_1}R$ denotes the $r$-fold covariant derivative
of the curvature $R$
with respect to $\nabla$ and some auxiliary connection $\overline{\nabla}$ on $\mS M$ in a
neighbourhood of $y$. This derivative is defined analogous to the classical (non-super) case
with appropriate signs. For $r=1,2$, it reads as follows.

\begin{Def}
\label{defCurvDerivatives}
Let $R\in\scal[\Hom_{\mO_{S\times M}}]{\mS M_S\otimes_{\mO_{S\times M}}\mS M_S\otimes_{\mO_{S\times M}}\mE_S}{\mE_S}$, and $u,v\in\mS M_S$. For $X,Y\in\mS M_S$, we define
\begin{align*}
\overline{\nabla}_X\scal[R]{u}{v}&:=\nabla_X\circ\scal[R]{u}{v}-(-1)^{\abs{R}\abs{X}}\scal[R]{\overline{\nabla}_Xu}{v}\\
&\qquad-(-1)^{\abs{X}(\abs{R}+\abs{u})}\scal[R]{u}{\overline{\nabla}_Xv}\\
&\qquad-(-1)^{\abs{X}(\abs{R}+\abs{u}+\abs{v})}\scal[R]{u}{v}\circ\nabla_X\\
\overline{\nabla}^2_{X,Y}\scal[R]{u}{v}&:=\overline{\nabla}_X\left(\overline{\nabla}_YR\right)\scal[(]{u}{v}
-\nabla_{\overline{\nabla}_XY}\circ\scal[R]{u}{v}\\
&\qquad+(-1)^{(\abs{X}+\abs{Y})\abs{R}}\scal[R]{\overline{\nabla}_{\overline{\nabla}_XY}u}{v}\\
&\qquad+(-1)^{(\abs{X}+\abs{Y})(\abs{R}+\abs{u})}\scal[R]{u}{\overline{\nabla}_{\overline{\nabla}_XY}v}\\
&\qquad+(-1)^{(\abs{X}+\abs{Y})(\abs{R}+\abs{u}+\abs{v})}\scal[R]{u}{v}\circ\nabla_{\overline{\nabla}_XY}
\end{align*}
\end{Def}

According to Exp. \ref{expGalaev}, the functor $\Hol_x$ is, in general, not
representable such that Galaev's holonomy theory is a priori different from ours.
Nevertheless, we will show that the generators of
$\hol^{\mathrm{Gal}}_x$ can be extracted in a geometric way, in a sense to be made precise.

\subsection{Parallel Transport and Covariant Derivatives}

The aforementioned extraction of generators of $\hol_x^{\mathrm{Gal}}$ is based on the
following observation. Consider again the more general situation of an $S$-connection $\nabla$
on $\mE_S$ for $S=\bR^{0|L}$ and $x:S\rightarrow M$ an $S$-point.
As shown next, the pullback connection $x^*\nabla$ - along with its
induced connections on tensors as well as higher covariant derivatives - arises by means
of infinitesimal parallel transport. We will not treat the most general situation here
but content ourselves with the following. First, we consider only even vector fields
to be differentiated along.
The general case is expected to work along the lines of the flow of vector bundles developed
in \cite{MSV93}.
Second, we consider tensors of the following type: sections, endomorphisms and curvature-type. The general case should be analogous. Third, we consider covariant derivatives up to second order. Analogous results for higher order derivatives are expected to be obtainable by an inductive proof.

For $X\in\mS S$ and $Z\in x^*\mE$, the pullback $(x^*\nabla)_XZ\in x^*\mE$ was defined
in (\ref{eqnPullbackConnectionRelative}). Let also $Y\in\mS S$ and
$\overline{\nabla}$ be an $S$-connection on $\mS M_S$. We define the second
covariant derivative of $Z$, with respect to $\nabla$ and $\overline{\nabla}$, as follows.
\begin{align}
(x^*\overline{\nabla}^2)_{X,Y}Z:=(x^*\nabla)_X(x^*\nabla)_YZ-(x^*\nabla)_{(x^*\overline{\nabla})_X[dx[Y]]}Z
\end{align}
The corresponding first and second covariant derivatives of endomorphisms
and tensors of curvature type are defined likewise.

\begin{Def}
\label{defEndoPullbackDerivatives}
Let $E\in\End_{\mO_{S\times M}}(\mE_S)$ be an endomorphism and $E_x$ its pullback
under $x$ as in (\ref{eqnEndoPullback}). For $X,Y\in\mS S$, we define
\begin{align*}
(x^*\nabla)_XE_x&:=(x^*\nabla)_X\circ E_x-(-1)^{\abs{X}\abs{E}}E_x\circ(x^*\nabla)_X\in\End_{\mO_S}(x^*\mE)\\
(x^*\overline{\nabla}^2)_{X,Y}E_x&:=(x^*\nabla)_X\left((x^*\nabla)_YE_x\right)\\
&\qquad-(x^*\nabla)_{(x^*\overline{\nabla})_X[dx[Y]]}\circ E_x+(-1)^{\abs{E}(\abs{X}+\abs{Y})}E_x\circ(x^*\nabla)_{(x^*\overline{\nabla})_X[dx[Y]]}
\end{align*}
\end{Def}

\begin{Def}
\label{defCurvPullbackDerivatives}
Let $R\in\scal[\Hom_{\mO_{S\times M}}]{\mS M_S\otimes_{\mO_{S\times M}}\mS M_S\otimes_{\mO_{S\times M}}\mE_S}{\mE_S}$, and $u,v\in x^*\mS M$. For $X,Y\in\mS S$, we define
\begin{align*}
(x^*\overline{\nabla})_X\scal[R_x]{u}{v}&:=(x^*\nabla)_X\circ\scal[R_x]{u}{v}-(-1)^{\abs{R}\abs{X}}\scal[R_x]{(x^*\overline{\nabla})_X(u)}{v}\\
&\qquad-(-1)^{\abs{X}(\abs{R}+\abs{u})}\scal[R_x]{u}{(x^*\overline{\nabla})_X(v)}\\
&\qquad-(-1)^{\abs{X}(\abs{R}+\abs{u}+\abs{v})}\scal[R_x]{u}{v}\circ(x^*\nabla)_X\\
(x^*\overline{\nabla}^2)_{X,Y}\scal[R_x]{u}{v}&:=(x^*\overline{\nabla})_X\left((x^*\overline{\nabla})_YR_x\right)\scal[(]{u}{v}
-(x^*\nabla)_{(x^*\overline{\nabla})_X[dx[Y]]}\circ\scal[R_x]{u}{v}\\
&\qquad+(-1)^{(\abs{X}+\abs{Y})\abs{R}}\scal[R_x]{(x^*\overline{\nabla})_{(x^*\overline{\nabla})_X[dx[Y]]}(u)}{v}\\
&\qquad+(-1)^{(\abs{X}+\abs{Y})(\abs{R}+\abs{u})}\scal[R_x]{u}{(x^*\overline{\nabla})_{(x^*\overline{\nabla})_X[dx[Y]]}(v)}\\
&\qquad+(-1)^{(\abs{X}+\abs{Y})(\abs{R}+\abs{u}+\abs{v})}\scal[R_x]{u}{v}\circ(x^*\nabla)_{(x^*\overline{\nabla})_X[dx[Y]]}
\end{align*}
\end{Def}

Our next lemma ensures existence of an $S$-path as occurring in the subsequent proposition
concerning first covariant derivatives.

\begin{Lem}
\label{lemSpecialPath}
Let $x:S\rightarrow M$ be an $S$-point and $\xi\in(x^*\mS M)_{\overline{0}}$.
We write (in coordinates around $x_0$) $\xi=(x^*\partial_i)\cdot\xi^i$ and assume that
every $\xi^i\in\mO_S$ is nilpotent.
Then there is an $S$-path $\gamma$ (connecting $x$ to some other $S$-point $y$) such that
$\ev|_0\partial_t\circ\gamma^{\sharp}=\xi$.
\end{Lem}

\begin{proof}
Through Def. \ref{defCanonicalInclusion}, and setting $x^{\sharp}(t):=t$, we extend $x$
to a map $x:S\times\bR\rightarrow S\times M\times\bR$. In this sense, we define
\begin{align*}
\gamma^{\sharp}:=x^{\sharp}\circ\sum_{n=0}^{\infty}\frac{\left(\sum_i(t\xi^i\partial_i)\right)^n}{n!}
\end{align*}
Every $\xi^i\partial_i$ is, by assumption, even and nilpotent such that there are no ordering problems and the sum is finite.
Such $\gamma$ is indeed a morphism by the derivation property of $\sum_i(t\xi^i\partial_i)$ as shown analogous
as in the proof of Lem. 1.1 in \cite{Hel08}. A straightforward calculation shows, moreover,
that $\gamma$ indeed satisfies the required initial condition.
\end{proof}

\begin{Prp}
\label{prpParallelTransportConnection}
Let $x:S\rightarrow M$ be an $S$-point, $Y\in\mE_S$ and $\xi\in(x^*\mS M)_{\overline{0}}$. Let
$\gamma$ be an $S$-path (connecting $x$ to some $y$) such that $\ev|_0\partial_t\circ\gamma^{\sharp}=\xi$. Then
\begin{align*}
\frac{d}{dt}|_0\left(P_{\gamma}|_{[0,t]}^{-1}(\gamma^*Y)\right)
=(x^*\nabla)_{\xi}(x^*Y)
\end{align*}
In particular, for $\xi=X\circ x^{\sharp}=dx[X]$ with $X\in(\mS S)_{\overline{0}}$, we find
\begin{align*}
\frac{d}{dt}|_0\left(P_{\gamma}|_{[0,t]}^{-1}(\gamma^*Y)\right)=(x^*\nabla)_X(x^*Y)
\end{align*}
Similarly, the first covariant derivatives of $E_x$ and $R_x$, with $E$ and $R$ as in Def. \ref{defEndoPullbackDerivatives} and Def. \ref{defCurvPullbackDerivatives}, arise from parallel transport as
\begin{align*}
\frac{d}{dt}|_0\left(P_{\gamma}|_{[0,t]}^{-1}\circ E_{\gamma}\circ P_{\gamma}|_{[0,t]}\right)
&=(x^*\nabla)_XE_x\\
\frac{d}{dt}|_0\left(P_{\gamma}|_{[0,t]}^{-1}\circ\scal[R_{\gamma}]{\oP_{\gamma}|_{[0,t]}(u)}{\oP_{\gamma}|_{[0,t]}(v)}\circ P_{\gamma}|_{[0,t]}\right)
&=\scal[(x^*\overline{\nabla})_XR_x]{u}{v}
\end{align*}
\end{Prp}

\begin{proof}
Let $(T^j)$ be an $\mE$-basis in a neighbourhood of $x_0(0)\in M_0$.
For $t$ sufficiently small, we identify $P_{\gamma|_{[0,t]}}$ and its inverse with
a matrix with respect to bases $(x^*T^j)$ and $(\gamma_t^*T^j)$. By (\ref{eqnParallel}), we find that
\begin{align*}
\ev|_{t=0}P_{\gamma|_{[0,t]}}=\id\;,\qquad\partial_t|_0P_{\gamma|_{[0,t]}}=-\mB(0)\;,\qquad
\partial_t|_0P_{\gamma|_{[0,t]}}^{-1}=\mB(0)
\end{align*}
where the sign in the last equation is due to replacing $t$ by $1-t$ in $\gamma^{-1}$ within
the definition of $\mB(t)$. The first statement is shown by the following calculation,
writing $Y=T^kY^k$.
\begin{align*}
\frac{d}{dt}|_0\left(P_{\gamma}|_{[0,t]}^{-1}(\gamma^*Y)\right)
=\mB(0)\cdot(x^*Y)+(x^*T^k)\partial_t|_0\gamma^*Y^k=(x^*\nabla)_{\xi}(x^*Y)
\end{align*}
For the second statement note that, by (\ref{eqnEndoPullback}), the matrix
of $E_{\gamma}$ is the pullback under $\gamma$ of the matrix of $E$. For $Y\in x^*\mE$,
we thus yield
\begin{align*}
\frac{d}{dt}|_0\left(P_{\gamma}|_{[0,t]}^{-1}\circ E_{\gamma}\circ P_{\gamma}|_{[0,t]}\right)(Y)
&=\left(\mB(0)E_x+\partial_t|_0E_{\gamma}-E_x\mB(0)\right)(Y)\\
&=\mB(0)E_x[Y]+X(E_xY)-E_x[X(Y)]-E_x\mB(0)[Y]\\
&=\left((x^*\nabla)_X\circ E_x-E_x\circ(x^*\nabla)_X\right)(Y)
\end{align*}
Finally, the third statement is established by an analogous calculation.
\end{proof}

We now come to second covariant derivatives. Let $X,Y\in(\mS S)_{\overline{0}}$ and consider a map
$\gamma:S\times[0,1]\times[0,1]\rightarrow M$ such that
\begin{align}
\label{eqnParallelHomotopy}
\ev|_{(0,0)}\partial_t\circ\gamma^*=X\circ x^*\;,\qquad
\ev|_{s=0}\partial_s\circ\gamma^{\sharp}=\overline{P}_{\gamma_{s=0}|_{[0,t]}}(Y\circ x^{\sharp})=:Y_t
\end{align}
such that $Y_0=Y\circ x^{\sharp}$. Such a homotopy indeed exists.
First, by Lem. \ref{lemSpecialPath}, there is $\tilde{\gamma}:S\times[0,1]\rightarrow M$ (parameter $t$) such that
the first condition in (\ref{eqnParallelHomotopy}) is satisfied. Now fix $t$. For this $t$, there is,
by the same lemma, an $S$-path $\gamma_t:S\times[0,1]\rightarrow M$ (parameter $s$) such that also the second condition holds true with parallel transport $\overline{P}_{\tilde{\gamma}}|_{[0,t]}$
on the right hand side.
By construction, $\gamma_t$ depends smoothly on $t$ and $s$, thus yielding $\gamma$ as required.

\begin{Prp}
\label{prpSecondCovariantDerivative}
Let $Z\in\mE_S$ and $E\in\End_{\mO_N}(\mE)$. Then
\begin{align*}
\frac{d}{dt}|_0\frac{d}{ds}|_0(P^2_{s,t})^{-1}(\gamma^*Z)&=(x^*\overline{\nabla}^2)_{X,Y}(x^*Z)\\
\frac{d}{dt}|_0\frac{d}{ds}|_0\left((P^2_{s,t})^{-1}\circ E_{\gamma}\circ P^2_{s,t}\right)
&=(x^*\overline{\nabla}^2)_{X,Y}E_x\\
\frac{d}{dt}|_0\frac{d}{ds}|_0\left((P^2_{s,t})^{-1}\circ\scal[R_{\gamma}]{\overline{P}^2_{s,t}(u)}{\overline{P}^2_{s,t}(v)}\circ P^2_{s,t}\right)
&=(x^*\overline{\nabla}^2)_{X,Y}R_x(u,v)
\end{align*}
with
\begin{align*}
P^2_{s,t}:=P_{\gamma_{t}|_{[0,s]}}\circ P_{\gamma_{s=0}|_{[0,t]}}\;,\qquad
\overline{P}^2_{s,t}:=\overline{P}_{\gamma_{t}|_{[0,s]}}\circ\overline{P}_{\gamma_{s=0}|_{[0,t]}}
\end{align*}
\end{Prp}

\begin{proof}
Using Prp. \ref{prpParallelTransportConnection} and
$\abs{Y_t^l}=\abs{dx[Y]^l}=\abs{Y(x^*(\xi^l)}=\abs{\xi^l}$, we calculate
\begin{align*}
&\partial_s|_0\partial_t|_0(P_{\gamma_{s=0}|_{[0,t]}})^{-1}(P_{\gamma_{t}|_{[0,s]}})^{-1}(\gamma^*Z)\\
&\qquad=\partial_t|_0(P_{\gamma_{s=0}|_{[0,t]}})^{-1}\partial_s|_0(P_{\gamma_{t}|_{[0,s]}})^{-1}(\gamma^*Z)\\
&\qquad=\partial_t|_0(P_{\gamma_{s=0}|_{[0,t]}})^{-1}\left((\gamma_t^*\nabla)_{Y_t}(\gamma_t^*Z)\right)\\
&\qquad=(-1)^{\abs{\xi^l}\abs{Z}}\partial_t|_0(P_{\gamma_{s=0}|_{[0,t]}})^{-1}\gamma_t^*(\nabla_{\partial_l}Z)\cdot dx[Y]^l
+(-1)^{\abs{\xi^l}\abs{Z}}x^*(\nabla_{\partial_l}Z)\partial_t|_0Y_t^l
\end{align*}
Now we use $\partial_t|_0(P_{\gamma_{s=0}|_{[0,t]}})^{-1}\gamma_t^*(\nabla_{\partial_l}Z)=(x^*\nabla)_X(x^*\nabla_{\partial_l}Z)$ and
\begin{align*}
\partial_t|_0Y_t=-\mB^X(s=0)dx[Y]
=-(-1)^{\abs{n}}dx[Y]^ndx[X](\xi^l)x^*(\overline{\nabla}_{\partial_{\xi^l}}\partial_{\xi^n})
\end{align*}
to yield the first statement after a straightforward calculation.

The left hand side of the second equation is treated as follows.
\begin{align*}
LHS&=\frac{d}{dt}|_0\left(P_{\gamma}|_{s=0,[0,t]}^{-1}\partial_s|_0\left(P_{\gamma}|_{t,[0,s]}^{-1}\circ E_{\gamma}\circ P_{\gamma}|_{t,[0,s]}\right)\circ P_{\gamma}|_{s=0,[0,t]}\right)\\
&=\frac{d}{dt}|_0\left(P_{\gamma}|_{s=0,[0,t]}^{-1}\left((\gamma_t^*\nabla)_{Y_t}\circ E_{\gamma_t}-E_{\gamma_t}\circ(\gamma_t^*\nabla)_{Y_t}\right)\circ P_{\gamma}|_{s=0,[0,t]}\right)\\
&=(x^*\nabla)_X\left((x^*\nabla)_YE_x\right)
-(x^*\nabla)_{(x^*\overline{\nabla})_X[dx[Y]]}\circ E_x+E_x\circ(x^*\nabla)_{(x^*\overline{\nabla})_X[dx[Y]]}
\end{align*}
Here, the second equation follows from Prp. \ref{prpParallelTransportConnection}
applied to the second condition in (\ref{eqnParallelHomotopy}).
For the third equation, we use again Prp. \ref{prpParallelTransportConnection} to obtain
the first term and find, in addition, two derivative terms
with respect to $(\gamma_t^*\nabla)_{Y_t}$ which are obtained as in the previous calculation.

Similarly we yield, for the left hand side of the last equation to be shown,
\begin{align*}
LHS&=\frac{d}{dt}|_0P_{\gamma}|_{s=0,[0,t]}^{-1}\left((\gamma_t^*\nabla)_{Y_t}\circ\scal[R_{\gamma_t}]{\overline{P}_{\gamma}|_{s=0,[0,t]}(u)}{\overline{P}_{\gamma}|_{s=0,[0,t]}(v)}\right.\\
&\qquad\qquad\qquad-\scal[R_{\gamma_t}]{(\gamma_t^*\overline{\nabla})_{Y_t}(\overline{P}_{\gamma}|_{s=0,[0,t]}(u))}{\overline{P}_{\gamma}|_{s=0,[0,t]}(v)}\\
&\qquad\qquad\qquad-\scal[R_{\gamma_t}]{\overline{P}_{\gamma}|_{s=0,[0,t]}(u)}{(\gamma_t^*\overline{\nabla})_{Y_t}(\overline{P}_{\gamma}|_{s=0,[0,t]}(v))}\\
&\qquad\qquad\qquad\left.-\scal[R_{\gamma_t}]{\overline{P}_{\gamma}|_{s=0,[0,t]}(u)}{\overline{P}_{\gamma}|_{s=0,[0,t]}(v)}\circ(\gamma_t^*\nabla)_{Y_t}\right)P_{\gamma}|_{s=0,[0,t]}
\end{align*}
Analogous to the previous calculation for the second statement,
Prp. \ref{prpParallelTransportConnection} together with derivative terms
from the first calculation yields the right hand side as claimed.
\end{proof}

\subsection{Reconstruction of Galaev's Holonomy Algebra}

By means of the previously established relation between covariant derivatives and parallel
transport, we will now make contact with Galaev's holonomy algebra $\hol_x^{\mathrm{Gal}}$.
Let $S=\bR^{0|0}$, $\nabla$ be a connection on $\mE\rightarrow M$, and $x\in M_0$ be a
topological point identified with an $S$-point.
We aim at gaining generating elements of $\hol_x^{\mathrm{Gal}}$ as coefficients of special
elements of $\hol_x(T)$ for $T=\bR^{0|L'}$ with $L'\geq(\dim M)_{\overline{1}}$. Let $q\in M_0$,
and define the $(S\times)T$-point $y$ by prescribing
\begin{align}
\label{eqnGalaevY}
y^{\sharp}(x^k):=q^*(x^k)=q^k\;,\qquad y^{\sharp}(\theta^i):=\eta^i
\end{align}
with respect to coordinates $\xi=(x,\theta)$ around $q$. Then, a straightforward calculation
using (\ref{eqnPullbackConnection}) shows that
\begin{align*}
(y^*\nabla)_{\partial_{\eta^j}}(y^*Z)&=\hat{y}^*(\nabla_{\partial_{\theta^j}}Z)\\
(y^*\nabla)_{(y^*\overline{\nabla})_{\partial_{\eta^j}}[dy[\partial_{\eta^k}]]}(y^*Z)
&=\hat{y}^*\left(\nabla_{\overline{\nabla}_{\partial_{\theta^j}}\partial_{\theta^k}}Z\right)
\end{align*}
For the curvature terms, it follows that
\begin{align}
\label{eqnGalaevTerms}
\scal[R_y]{y^*\partial_{\xi^i}}{y^*\partial_{\xi^j}}
&=\hat{y}^*\left(\scal[R]{\partial_{\xi^i}}{\partial_{\xi^j}}\right)\nonumber\\
\scal[\left((y^*\overline{\nabla})_{\partial_{\eta^l}}R_y\right)]{y^*\partial_{\xi^i}}{y^*\partial_{\xi^j}}
&=\hat{y}^*\left(\scal[(\overline{\nabla}_{\partial_{\theta^l}}R)]{\partial_{\xi^i}}{\partial_{\xi^j}}\right)\\
\scal[\left((y^*\overline{\nabla}^2)_{\partial_{\eta^l},\partial_{\eta^m}}R_y\right)]{y^*\partial_{\xi^i}}{y^*\partial_{\xi^j}}
&=\hat{y}^*\left(\scal[(\overline{\nabla}^2_{\partial_{\theta^l},\partial_{\theta^m}}R)]{\partial_{\xi^i}}{\partial_{\xi^j}}\right)\nonumber
\end{align}

\begin{Lem}
\label{lemGalaevTerms}
Let $y$ be the $T$-point (\ref{eqnGalaevY}), $\gamma:x\rightarrow y$ be a connecting $T$-path
and $I_k$ denote a multiindex of parity $\abs{\xi^k}$ such that $\eta^{I_k}\in\mO_T$. Then
\begin{align*}
\eta^{I_{k_1}}\eta^{I_{k_2}}\cdot P_{\gamma}^{-1}\circ y^*\left(\scal[R]{\partial_{\xi^{k_2}}}{\partial_{\xi^{k_1}}}\right)\circ P_{\gamma}&\in\hol_x(T)\\
\eta^{I_{k_1}}\eta^{I_{k_2}}\eta^{I_{k_3}}\cdot P_{\gamma}^{-1}\circ y^*\left(\scal[(\nabla_{\partial_{\theta^{k_3}}}R)]{\partial_{\xi^{k_2}}}{\partial_{\xi^{k_1}}}\right)\circ P_{\gamma}&\in\hol_x(T)\\
\eta^{I_{k_1}}\eta^{I_{k_2}}\eta^{I_{k_3}}\eta^{I_{k_4}}\cdot P_{\gamma}^{-1}\circ y^*\left(\scal[(\nabla^2_{\partial_{\theta^{k_4}},\partial_{\theta^{k_3}}}R)]{\partial_{\xi^{k_2}}}{\partial_{\xi^{k_1}}}\right)\circ P_{\gamma}&\in\hol_x(T)
\end{align*}
\end{Lem}

\begin{proof}
By Thm. \ref{thmAmbroseSinger}, the first term
\begin{align*}
\eta^{I_{k_1}}\eta^{I_{k_2}}\cdot P_{\gamma}^{-1}\circ y^*\left(\scal[R]{\partial_{k_2}}{\partial_{k_1}}\right)\circ P_{\gamma}
=P_{\gamma}^{-1}\circ\scal[R_y]{\eta^{I_{k_2}}\cdot(y^*\circ\partial_{k_2})}{\eta^{I_{k_1}}\cdot(y^*\circ\partial_{k_1})}\circ P_{\gamma}
\end{align*}
is clearly contained in $\hol_x(T)$.
For the second, let $\delta$ be an $S$-path connecting $y$ to some $S$-point $z$ such that
$\ev|_0\partial_t\circ\delta^{\sharp}=\xi:=dy\left[\eta^{I_{k_3}}\cdot\partial_{\eta^{k_3}}\right]$.
Using (\ref{eqnGalaevTerms}), followed by Prp. \ref{prpParallelTransportConnection}
applied to $y$, $\xi$, $\delta$ as well as $u:=\eta^{I_{k_2}}\cdot(y^*\circ\partial_{k_2})$ and
$v:=\eta^{I_{k_1}}\cdot(y^*\circ\partial_{k_1})$, we yield
\begin{align*}
&\eta^{I_{k_1}}\eta^{I_{k_2}}\eta^{I_{k_3}}\cdot P_{\gamma}^{-1}\circ y^*\left(\scal[(\nabla_{\theta^{k_3}}R)]{\partial_{k_2}}{\partial_{k_1}}\right)\circ P_{\gamma}\\
&\qquad=P_{\gamma}^{-1}\circ\scal[\left((y^*\nabla)_{\eta^{I_{k_3}}\cdot\partial_{\eta^{k_3}}}R_y\right)]{\eta^{I_{k_2}}\cdot(y^*\circ\partial_{k_2})}{\eta^{I_{k_1}}\cdot(y^*\circ\partial_{k_1})}\circ P_{\gamma}\\
&\qquad=P_{\gamma}^{-1}\circ\partial_t|_0\left(P_{\delta}|_{[0,t]}^{-1}\circ\scal[R_{\delta}]{\oP_{\delta}|_{[0,t]}(u)}{\oP_{\delta}|_{[0,t]}(v)}\circ P_{\delta}|_{[0,t]}\right)\circ P_{\gamma}\\
&\qquad=\partial_t|_0\left(P_{\gamma}^{-1}\circ P_{\delta}|_{[0,t]}^{-1}\circ\scal[R_{\delta}]{\oP_{\delta}|_{[0,t]}(u)}{\oP_{\delta}|_{[0,t]}(v)}\circ P_{\delta}|_{[0,t]}\circ P_{\gamma}\right)
\end{align*}
By Thm. \ref{thmAmbroseSinger}, the term in parentheses lies, for every $t\in[0,1]$,
in $\hol_x(T)$, which is a vector space.
Therefore, the differential is also contained in $\hol_x(T)$.

The second covariant derivative term is treated analogous.
\end{proof}

Consider the zero-derivative term in Lem. \ref{lemGalaevTerms}. For generic choice of
$\eta^{I_{k_1}}$ and $\eta^{I_{k_2}}$, we find that
\begin{align*}
&\left(\partial_{\eta^{I_{k_1}}}\partial_{\eta^{I_{k_2}}}\left(\eta^{I_{k_1}}\eta^{I_{k_2}}P_{\gamma}^{-1}\circ y^*\left(\scal[R]{\partial_{\xi^{k_2}}}{\partial_{\xi^{k_1}}}\right)\circ P_{\gamma}\right)\right)_0\\
&\qquad=P_{\gamma_0}^{-1}\circ\scal[R_{y_0}]{\partial_{\xi^{k_2}}}{\partial_{\xi^{k_1}}}\circ P_{\gamma_0}
\in\hol^{\mathrm{Gal}}_x
\end{align*}
and analogous for the first and second derivative terms and, by conjecture, for all higher
derivative terms. The generating elements of $\hol_x^{\mathrm{Gal}}$ can thus be
extracted out of $\hol_x(T)$ as certain coefficients of special elements in the
way made precise by Lem. \ref{lemGalaevTerms}.
This construction is based on the knowledge of the geometric significance of the
elements. It remains an open question whether $\hol_x^{\mathrm{Gal}}$ can be obtained
from $\hol_x(T)$ in a purely algebraic way.

\section*{Acknowledgements}

I would like to thank Alexander Alldridge, Anton Galaev, Harald Dorn and Jan Plefka for interesting discussions, and again Anton Galaev for spotting a mistake in a previous version of this article.

\bibliographystyle{alpha}

\end{document}